\documentclass[letterpaper,11pt]{article}

\usepackage[font=small,labelfont=bf]{caption}
\usepackage[utf8]{inputenc}
\usepackage{microtype}
\usepackage{xcolor}
\usepackage{amsthm}
\usepackage{amsmath}
\usepackage{amsthm}
\usepackage{amssymb}
\usepackage{xspace}
\usepackage{graphicx}
\usepackage{fullpage}
\usepackage[colorlinks=false,backref=page]{hyperref}
\usepackage{url}
\usepackage[capitalise]{cleveref}
\usepackage[shortlabels]{enumitem}
\usepackage{subfig}
\usepackage[export]{adjustbox}
\usepackage{thm-restate}

\newtheorem{theorem}{Theorem}[section]
\newtheorem{lemma}[theorem]{Lemma}

\newtheorem{definition}[theorem]{Definition}

\newtheorem{corollary}[theorem]{Corollary}

\newtheorem{observation}[theorem]{Observation}

\newcommand{\tw}{\mathrm{\textbf{tw}}}

\newcommand{\br}[1]{\left(#1\right)}

\newcommand{\eps}{\varepsilon}
\newcommand{\ex}[1]{\mathbb{E}\left[#1\right]}
\newcommand{\pr}[1]{\mathbb{P}\left[#1\right]}
\newcommand{\Oh}{\mathcal{O}}
\newcommand{\poly}{\ensuremath{\text{poly}}}
\newcommand{\opt}{\ensuremath{\mathsf{opt}}}
\newcommand{\sub}{\subseteq}
\newcommand{\sm}{\setminus}

\newcommand{\rr}{\ensuremath{\mathbb{R}}}
\newcommand{\nn}{\ensuremath{\mathbb{N}}}

\newcommand{\gbf}{\ensuremath{\mathbf{G}}}
\newcommand{\fcal}{\ensuremath{\mathcal{F}}}
\newcommand{\acal}{\ensuremath{\mathcal{A}}}
\newcommand{\bcal}{\ensuremath{\mathcal{B}}}
\newcommand{\rcal}{\ensuremath{\mathcal{R}}}
\newcommand{\pcal}{\ensuremath{\mathcal{P}}}
\newcommand{\scal}{\ensuremath{\mathcal{S}}}
\newcommand{\shat}{\ensuremath{{\mathcal{S}}}}

\newcommand{\twdel}{\textsc{Weighted Treewidth-$\eta$ Deletion}\xspace}
\newcommand{\planardel}{\textsc{Weighted Planar $\mathcal{F}$-M-Deletion}\xspace}
\newcommand{\fdel}{\textsc{Weighted $\mathcal{F}$-M-Deletion}\xspace}
\newcommand{\ufdel}{\textsc{$\mathcal{F}$-M-Deletion}\xspace}
\newcommand{\fvs}{\textsc{Feedback Vertex Set}\xspace}
\newcommand{\wei}{\textsc{Weighted}\xspace}
\newcommand{\cover}{-modulator hitting family\xspace}

\def\DEBUG{true}

\ifdefined\DEBUG{}

\def\rem#1{{\marginpar{\raggedright\scriptsize #1}}}
\newcommand{\micr}[1]{\rem{\textcolor{blue}{\(\bullet \) #1}}}

\newcommand{\meir}[1]{\rem{\textcolor{purple}{\(\bullet \) #1}}}
\else

\newcommand{\micr}[1]{ }
\newcommand{\meir}[1]{ }
\fi

\title{Losing Treewidth In The Presence Of Weights\footnote{This project has been supported by Polish National Science Centre SONATA-19 grant 2023/51/D/ST6/00155.}}
\author{Michał Włodarczyk \\ University of Warsaw \\  {\texttt{michal.wloda@gmail.com}}}
\date{}
\begin{document}

\maketitle

\begin{abstract}
   In the \twdel problem we are given a node-weighted graph $G$ and we look for a vertex subset $X$ of minimum weight such that the treewidth of $G-X$ is at most~$\eta$.
   We show that \twdel admits a randomized polynomial-time constant-factor approximation algorithm for every fixed $\eta$.
   Our algorithm also works for the more general \planardel problem.
   %We also give a randomized single-exponential FPT algorithm parameterized by the number of vertices in the solution.
   %The approximation factor and the polynomial degree are computable functions of $\eta$.

   %As a consequence, we obtain a constant-factor approximation for \planardel as well.
   This work extends the results for unweighted graphs by [Fomin,  Lokshtanov,  Misra, Saurabh; FOCS '12] and answers a question posed by [Agrawal, Lokshtanov, Misra, Saurabh, Zehavi; APPROX/RANDOM '18] and [Kim, Lee, Thilikos; APPROX/RANDOM '21].
    The presented algorithm is based on a novel technique of random sampling of so-called protrusions.
\end{abstract}

\iffalse
In the Weighted Treewidth-Eta Deletion problem we are given a node-weighted graph G and we look for a vertex subset X of minimum weight such that the treewidth of G-X is at most Eta. We show that Weighted Treewidth-Eta Deletion admits a randomized polynomial-time constant-factor approximation algorithm for every fixed Eta. Our algorithm also works for the more general Weighted Planar F-M-Deletion problem.

 This work extends the results for unweighted graphs by [Fomin,  Lokshtanov,  Misra, Saurabh; FOCS '12] and answers a question posed by [Agrawal, Lokshtanov, Misra, Saurabh, Zehavi; APPROX/RANDOM '18] and [Kim, Lee, Thilikos; APPROX/RANDOM '21]. The presented algorithm is based on a novel technique of random sampling of so-called protrusions.
\fi

\section{Introduction}

Let $\mathcal{H}$ be a class of  graphs.
A natural vertex deletion problem corresponding to $\mathcal{H}$ is
defined as follows: given a graph $G$ find a set $X \sub V(G)$ of minimum size for which $G-X \in \mathcal{H}$.
%In the {\sc Weighted $\mathcal{H}$-Deletion}\micr{change name} problem we are given a graph $G$ with a weight function $w\colon V(G) \to \rr^+$ and we ask for a vertex set $X \sub V(G)$ of minimum total weight, for which $G-X \in \mathcal{H}$.
%In the unweighted variant all the weights are equal 1.
This captures numerous fundamental graph problems such as {\sc Vertex Cover}, \fvs, or {\sc Vertex Planarization}. %, or {\sc $k$-Separator}
Vertex deletion problems have been long studied from the perspective of approximation algorithms, starting from the classic 2-approximation for {\sc Vertex Cover}~\cite{bar1981linear, nemhauser1974properties} from the 70s.  
It is known that for every non-trivial hereditary class $\mathcal{H}$ the corresponding vertex deletion problem
is NP-hard and APX-hard~\cite{lund1993approximation} so the best we can hope for in polynomial time is a constant-factor approximation.
Alike for {\sc Vertex Cover}, this can be achieved for any class $\mathcal{H}$ characterized by a finite family of forbidden induced subgraphs via a linear programming relaxation (or sometimes by other means~\cite{guruswami2017inapproximability,Lee19,OKUN2003231}).
Lund and Yannakakis~\cite{lund1993approximation} conjectured in the early 90s that this condition in fact captures all vertex deletion problems admitting a constant-factor approximation.
This was quickly disproved with a 2-approximation for \fvs~\cite{bafna19992, becker1994approximation} as the class of acyclic graphs (i.e. forests) does not admit such a characterization.
Later, Fujito~\cite{fujito1997primal} constructed an infinite family of counterexamples based on matroidal properties. % $\mathcal{H}$ for which {\sc Weighted $\mathcal{H}$-Deletion} admits a constant-factor approximation.
Nowadays, we know several more interesting counterexamples: block graphs~\cite{agrawal2016faster}, 3-leaf power graphs~\cite{ahn2023polynomial}, (proper) interval graphs~\cite{cao2016linear, HofV13}, ptolemaic graphs~\cite{ahn2022towards}, and bounded-treewidth graphs~\cite{FominLMS12, GuptaLLM019}.

The landscape of positive results is grimmer when we consider {\em weighted vertex deletion} problems.
Here, the graph $G$ is equipped with a  weight function $w\colon V(G) \to \rr^+$ and we ask for a vertex set $X \sub V(G)$ of minimum total weight for which $G-X \in \mathcal{H}$.
The constant-factor approximations for {\sc Vertex Cover}, \fvs, and for classes defined by a finite family of forbidden induced subgraphs were originally presented already in the weighted setting, but apart from that we are only aware of such algorithms for ptolemaic graphs~\cite{ahn2022towards}, $\theta_c$-free graphs~\cite{fiorini2010hitting, KimLT21} (a special case of bounded-treewidth graphs; explained shortly), and the matroidal classes introduced by Fujito~\cite{fujito1997primal}.

%in general such results are scarce and apply yo very specific graph classes~\cite{fiorini2010hitting,fujito1997primal,KimLT21}.

\paragraph{Treewidth reduction and minor hitting.}

The class of forests coincides with the class of graphs with treewidth one~\cite[\S 7]{cygan2015parameterized}.
Hence a natural generalization of (\wei) \fvs is (\wei) {\sc Treewidth-$\eta$ Deletion}, a problem studied in the context of bidimensionality theory~\cite{demaine2005bidimensionality, fomin2012bidimensionality, FominL0Z20, fomin2018excluded} and satisfiability~\cite{Bansal0U17}.
Here, we look for a minimum-weight {\em treewidth-$\eta$ modulator}, that is, a vertex set $X \sub V(G)$ for which the treewidth of $G-X$ is bounded~by~$\eta$.

Treewidth is closely related to the theory of graph minors (see \Cref{sec:prelim} for definitions). % and basic facts).
Since the graph property ``treewidth at most $\eta$'' is closed under taking minors, Robertson and
Seymour's Theorem implies that it can be characterized by a finite family of forbidden minors~\cite{GM20}.
In addition, a minor-closed graph class has bounded treewidth if and only if at least one of its forbidden minors is planar~\cite{GM5}.
The treewidth bound is known to be polynomial in the size of this planar~graph~\cite{chuzhoy2021towards}.

For a finite family of graphs $\fcal$, the (\wei) \ufdel problem (also known as {\sc $\fcal$-Minor Cover}) asks for a minimum-size (resp. minimum-weight) set $X \sub V(G)$ for which $G-X$ does not contain any graph $F \in \fcal$ as a \underline{minor}.
If the family $\fcal$ contains a planar graph, we speak of  (\wei) {\sc Planar} \ufdel.
Observe that  {\sc Treewidth-$\eta$ Deletion} is a special case of {\sc Planar} \ufdel. % for some family $\fcal$ containing a planar graph.
This special case is paramount to understand the general case: if $X \sub V(G)$ is a solution to {\sc Planar} \ufdel then $\tw(G-X) \le \eta$ for some constant $\eta$ depending on $\fcal$.
Hence if we were able to compute a constant-approximate treewidth-$\eta$ modulator $Y$, we could then solve the problem optimally on the bounded-treewidth graph $G-Y$, and output the union of both solutions.

Other examples of \ufdel include deletion problems into the classes of outerplanar graphs (for $\fcal = \{K_4,K_{3,2}\}$), planar graphs (for $\fcal = \{K_5,K_{3,3}\}$), or classes of bounded pathwidth/treedepth.
Note that in the second case the family $\fcal$ does not contain a planar graph, which corresponds to the fact that the class of planar graphs has unbounded treewidth.

\paragraph{Approximating \ufdel.}
In their seminal work, Fomin, Lokshtanov, Misra, and Saurabh~\cite{FominLMS12} presented a randomized constant-factor approximation for every {\sc Planar} \ufdel problem in the unweighted setting.
Subsequently, Gupta, Lee, Li, Manurangsi, and Włodarczyk~\cite{GuptaLLM019} gave a deterministic algorithm with explicit approximation factor $\Oh(\log \eta)$ where $\eta$ is the treewidth bound imposed by excluding the family~$\fcal$.
Obtaining a constant-factor approximation for any family $\fcal$ without a planar graph is wide open.
This includes the {\sc Vertex Planarization} problem ($\fcal = \{K_5,K_{3,3}\}$) for which 
the best approximation factors are $\Oh_\eps(\opt^{\eps})$ in polynomial time, for any $\eps > 0$, and $\log^{\Oh(1)}(\opt)$ in quasi-polynomial time~\cite{Jansen022, kawarabayashi2017polylogarithmic}.

The current-best approximation algorithms for \planardel have factors $\Oh(\log^{1.5} |V(G)|)$ (randomized), $\Oh(\log^{2} |V(G)|)$ (deterministic)~\cite{AgrawalLM0Z18} %\micr{they solved more problem, like chordal deletion} 
and $\Oh(\log |V(G)|\log\log |V(G)|)$ for the edge-deletion version~\cite{Bansal0U17}.
Apart from \wei\fvs, the only other case known to admit a constant-factor approximation corresponds to $\fcal = \{\theta_c\}$\footnote{The multigraph $\theta_c$ consists of two vertices joined by $c$ parallel edges. Excluding $\theta_c$ as a minor can be  expressed as excluding a finite family of simple graphs containing a planar graph, e.g., subdivided $\theta_c$.}~\cite{KimLT21}.
Beforehand, a specialized algorithm was given for $\fcal = \{\theta_3\}$ where the corresponding deletion problem was dubbed {\sc Weighted Diamond Hitting Set}~\cite{fiorini2010hitting}.
Apart from that, \planardel admits a PTAS on bounded-genus graphs and a~QPTAS on $H$-minor-free graphs\footnote{These results were presented only for \twdel or for families $\fcal$ with only connected graphs.
This is because they relied on the conference version of the work~\cite{BasteST20} which also required this restriction.
In the journal version of~\cite{BasteST20} this restriction has been dropped so the results in~\cite{FominL0Z20} can be automatically strengthened~\cite{meirav-personal}.}~\cite{FominL0Z20}.
To the best of our knowledge, the existence of a polynomial-time constant-factor approximation for \planardel remains open even on $H$-minor-free graphs.
%A common feature of all these algos is usage of LP techniques.
%\micr{add a sentence to make a nicer break}

\paragraph{Protrusion replacement.}
The results of Fomin et al.~\cite{FominLMS12} are based on the powerful technique of {\em protrusion replacement}~\cite{fellows1989analogue}.
An {\em $r$-protrusion} in a graph $G$ is a vertex subset $A \sub V(G)$ for which $\tw(G[A]) \le r$ and $|\partial_G(A)| \le r$, where $\partial_G(A)$ is the subset of vertices from $A$ that have neighbors outside $A$.
The idea of protrusion replacement is to analyze the behavior of an $r$-protrusion with respect to the considered problem and replace it with an equivalent subgraph of constant size.
Apart from constant-factor approximations, the same work  spawned single-exponential FPT algorithms parameterized by the solution size (under some additional restriction) and polynomial kernelizations (using a more advanced concept of a {\em near-protrusion}). %~\cite{FominLMS12}.
For other applications of this technique, see e.g.,~\cite{BasteST23hittingIV, bodlaender2016meta, FominLST20, Jansen022, LokshtanovR0Z18}.

The caveat of protrusion replacement is that in general it is only known to work for the unweighted problems.
Kim, Lee, and Thilikos~\cite{KimLT21} designed a weighted variant of this technique for the special case of \planardel with  $\fcal = \{\theta_c\}$.
They suggested a roadmap to achieve a constant-factor approximation for the general case, concluding that
``($\dots$) the most challenging step is
to design a weighted protrusion replacer''.

We remark that the techniques used by Gupta et al.~\cite{GuptaLLM019} in their constant-factor approximation for {\sc Planar} \ufdel also fall short to process weighted graphs.
%do not seem to easily generalize to the weighted setting.
Their algorithm repeatedly improves a solution $X$ by removing a small vertex set $Y$ to reduce the number of $X$-vertices in each connected component to $\Oh_\fcal(1)$. 
The argument ensuring the existence of such a set  $Y$ relies on the properties of separators in a bounded-treewidth graph
and breaks down when we replace ``small'' with ``of small weight''. 
%fact that removing any bag in a tree decomposition of small width is comparably costly, and breaks down in the presence of weights.\micr{properties of separators}
%The existence of a set $Y$ of small cardinality follows from the structure of separators in a bounded-treewidth graph but we could not control the weights in $Y$.

%\micr{ Lee \cite{Lee19}}
%suggested that obtaining an $\Oh(1)$-approximation for more complicated cases of \planardel may require designing a general weighted counterpart of protrusion replacement.

\paragraph{FPT algorithms for weighted problems.}
Apart from approximation, \ufdel has been studied extensively from the perspective of %parameterized algorithms
fixed-parameter tractability (FPT)~\cite{chen2010improved, cygan2012improved, fellows1988nonconstructive, jansen2014near, joret2014hitting, kim2012single, KociumakaP19, LiN22, MorelleSST23, SauST22}.
It is known that every {\sc Planar} \ufdel problem can be solved in time $2^{\Oh_\fcal(k)}\cdot n^{\Oh(1)}$ where $k$ stands for the solution size~\cite{kim2015linear}.
Such a single-exponential parameter dependency is best possible under the Exponential Time Hypothesis already
 for {\sc Vertex Cover}~\cite{impagliazzo2001problems}.

We will consider parameterized %{\em fixed-parameter tractable} (FPT) 
algorithms for weighted problems under the parameterization by the solution size $k$.
More precisely, we look for a solution of minimum weight among those on at most $k$ vertices, called a {\em $k$-optimal} solution (see, e.g.,~\cite{cygan2014minimum, EtscheidKMR17, fomin2016efficient, flow-dir, flow-undir, MandalMRS24, shachnai2017multivariate}). %, in running time of the form $f(k)\cdot n^{\Oh(1)}$.
%Even thought a large part of the FPT toolbox is tailored for unweighted problems, some techniques can be generalized to handle weights~\cite{cygan2014minimum, fomin2016efficient, shachnai2017multivariate}.
%A remarkable example is the framework of {\em flow augmentation}~\cite{flow-dir, flow-undir} that allowed to, e.g., classify {\sc Weighted Directed Feedback Vertex Set} as FPT.
Observe that the treewidth of a solvable instance of \planardel cannot exceed $k + \Oh_\fcal(1)$.
Therefore, we can find a solution in time $2^{\Oh_\fcal(k\log k)}\cdot n^{\Oh(1)}$ via dynamic programming on a tree decomposition~\cite{BasteST23hittingIV} (see \cite[\S 7]{SauST22} for a discussion on treating weights). 
A single-exponential FPT algorithm is known for \wei\fvs~\cite{agrawal2016faster, chen2008improved}.
Although we are not aware of any more general results,
it is plausible that the techniques developed by Kim et al.~\cite{kim2015linear} can be adapted to solve general \planardel in single-exponential time~\cite{ignasi-personal}.

\subsection{Our results}

Our main conceptual message is that the need to design a weighted counterpart of the protrusion replacement technique can be circumvented.
Instead, our algorithms are based on a new structural result about protrusions and treewidth modulators.
We say that a non-empty family of $r$-protrusions $\mathcal{P}$ in $G$ forms an {\em $(\eta, r, c)$\cover} if for every treewidth-$\eta$ modulator $X$ in $G$ at least $|\pcal|\, /\, c$ sets in $\mathcal{P}$ have a non-empty intersection with $X$.

\begin{restatable}{theorem}{thmCover}
\label{thm:main:cover}
    %For every $\eta \in \nn$ there is $c_\eta \in \nn$ such that 
    For each $\eta > 0$,
    every graph $G$ with $\tw(G) > \eta$ has an $\br{\eta, \Oh(\eta), 2^{\Oh(\eta)}}$\cover $\mathcal{P}$ of size $|V(G)|^{\Oh(\eta)}$. Furthermore, $\mathcal{P}$ can be constructed in time $|V(G)|^{\Oh(\eta)}$.%\micr{computability of $c$}
\end{restatable}

Equipped with a modulator hitting
family $\pcal$ we can randomly sample an $\Oh(\eta)$-protrusion $A \in \pcal$ and then guess a piece of an optimal solution $X \cap A$ by considering $\Oh_\eta(1)$ options.
The concrete choice of the probability distribution will be adjusted to specific applications. 

This strategy bears some resemblance to random sampling of important separators~\cite[\S 8]{cygan2015parameterized} used in the FPT algorithm for {\sc Multicut} by Marx and Razgon~\cite{MarxR14}.
Moreover, the proof of \Cref{thm:main:cover} also involves important separators and structures somewhat similar to {\em important clusters} from~\cite{MarxR14}. 
We remark however that the setups of these two techniques are disparate when it comes to the separator size (constant vs. parameter), the desired subset property (bounded treewidth vs. disjointedness with a given set) and the objective (intersect a solution vs. cover vertices that are in some sense irrelevant).
%o the best of our knowledge, important separators have not been employed so far in the context of {\sc Planar \fcal-Deletion}.\micr{approximation algos? but we used them in "losing treewidth"}
\Cref{thm:main:cover} can also be compared to the algorithm by Fomin el al.~\cite{FominLMPPS22} that involves random sampling of low-treewidth subgraphs in an apex-minor-free graph to cover an unknown small connected vertex set.
That construction however relies on totally different techniques. 

Our new tool enables us to design a first constant-factor approximation for \planardel, extending the results by Fomin et al.~\cite{FominLMS12} and Gupta et al.~\cite{GuptaLLM019} to the weighted realm.
%In particular, this captures \twdel and
This resolves an open problem by Agrawal et al.~\cite{AgrawalLM0Z18} and Kim et al.~\cite{KimLT21}.
In the statement below, $\eta(\fcal)$ denotes the treewidth bound for the class of graphs excluding all $F \in \fcal$ as minors.

\begin{restatable}{theorem}{thmApx}
\label{thm:main:apx}
    Let $\fcal$ be a finite family of graphs containing a planar graph. % and $\eta$ be the treewidth bound
    The \fdel problem admits a randomized %polynomial-time 
    constant-factor approximation algorithm with running time $\Oh_\fcal\br{ |V(G)|^{\Oh(\eta(\fcal))}}$. 
\end{restatable}

Notably, all the known approximation algorithms for \planardel~\cite{AgrawalLM0Z18, Bansal0U17, FominL0Z20}
and its special cases~\cite{bafna19992, becker1994approximation, fiorini2010hitting, fujito1997primal, KimLT21} rely on LP techniques\footnote{The first algorithms for \wei\fvs~\cite{bafna19992, becker1994approximation}
were originally phrased in a combinatorial language but later Chudak et al.~\cite{chudak1998primal} gave a simplified analysis of them in terms of 
the primal-dual method.} whereas our algorithm is purely combinatorial.
In particular, we do not follow the roadmap suggested by Kim et al.~\cite{KimLT21}.
Circumventing protrusion replacement helps us avoid another issue troubling that technique, namely non-constructivity.
Instead, our treatment of protrusions relies on explicit bounds on the sizes of minimal representative boundaried graphs in a certain equivalence relation by Baste, Sau, and Thilikos~\cite{BasteST23hittingIV}.
The values of the approximation factor 
and the constant hidden in the $\Oh_\fcal$-notation
can be retraced from that work but they are at least double-exponential in $\eta(\fcal)$.
%the maximal size of a graph in~$\fcal$.

As another application of \cref{thm:main:cover}, we present a single-exponential FPT algorithm for \planardel parameterized by the solution size $k$. %, improving upon the time complexity $2^{\Oh_\fcal(k\log k)}\cdot n^{\Oh(1)}$~\cite{BasteST23hittingIV}.
As noted before, such a result could possibly be achieved by combining the protrusion decomposition by Kim et at.~\cite{kim2015linear} with the exhaustive families described in \Cref{sec:exha}, without using \Cref{thm:main:cover}.
Nevertheless, we include it as an illustration of versatility of our approach.

\begin{restatable}{theorem}{thmFpt}
\label{thm:main:fpt}
    Let $\fcal$ be a finite family of graphs containing a 
    planar graph. % and $\eta$ be the treewidth bound 
    %There is a constant $c$ for which a $k$-optimal solution to \fdel %on an $n$-vertex graph 
    A $k$-optimal solution to
     \fdel
     can be computed in randomized time $2^{\Oh_\fcal(k)}\cdot |V(G)|^{\Oh(\eta(\fcal))}$.%$c^k\cdot n^{\Oh(\eta(\fcal))}$.\micr{$2^{\Oh(k)}$}
\end{restatable}

 This algorithm returns a $k$-optimal solution with constant probability; otherwise it may either output a solution of larger weight or report a failure.

\section{Techniques}
 \label{sec:outline}

We will now discuss how \Cref{thm:main:cover} facilitates designing approximation algorithms and compare our techniques to the previous works.

\paragraph{Protrusion replacement.}
The starting point is %of our investigation was 
the following randomized algorithm for unweighted \fvs.
First, we can perform a simple preprocessing to remove vertices of degree 1 and 2. %after which the minimum degree in a graph is at least 3.
Then the number of edges in an $n$-vertex graph $G$ is at least $\frac{3n}{2}$.
But when $X$ is a feedback vertex set then $G-X$ is a forest, which can have at most $n-1$ edges.
Therefore at least $\frac{1}{3}$-fraction of edges is incident to $X$.
We can thus sample a random edge and remove
one of its endpoints, also chosen at random.
The expected number of steps before
we hit a vertex from $X$ is at most 6 so this algorithm yields 6-approximation in expectation.

%\paragraph{Protrusion replacement.}
Fomin et al.~\cite{FominLMS12} invented a far-reaching generalization of this idea for the case when $X$ is a treewidth-$\eta$ modulator, i.e., $\tw(G-X) \le \eta$ for some constant $\eta \in \nn$.
Their preprocessing procedure does not merely remove vertices of low degree but it detects and simplifies %structures called 
$r$-protrusions.
A simple yet crucial observation is that the subgraph
 $G-X$ can be covered by $\Oh(|N(X)|)$ many $r$-protrusions for $r = 2\eta + 2$.
Hence if we could only ensure that each $r$-protrusion has size $\Oh_r(1)$ then again a constant fraction of edges would be incident to $X$ and we could apply the same argument as above.
This is what protrusion replacement is about: given an $r$-protrusion $A$ we would like to perform a ``surgery'' on $G$ to replace $G[A]$ with a subgraph of size $\Oh_r(1)$ that exhibits the same behavior.

The property of ``being $\fcal$-minor-free'' (as well as ``having bounded treewidth'') can be expressed in counting monadic
second-order (CMSO) logic.
For such graph properties, Courcelle's Theorem~\cite[Lemma 3.2]{bodlaender2016meta}
guarantees that subgraphs of constant-size boundary can be divided into a constant number of equivalence classes and replacing a subgraph with an equivalent one does not affect the property in question.
But this is not enough as we do not know which vertices from such a subgraph are being removed by a solution.
Protrusion replacement relies crucially on the observation that {\bf it is cheap to take into the solution the entire boundary of an $r$-protrusion.}
For a broad spectrum of problems this
property (formalized as {\em strong monotonicity}~\cite{bodlaender2016meta, FominLMS12})
implies that
 the difference in cost between any two solutions over an  $r$-protrusion is bounded by $\Oh_r(1)$.
Consequently,
there may be only finitely many non-equivalent ``behaviors'' of an $r$-protrusion.
For each such equivalence class there is some $r$-protrusion of minimal size and it can be safely used as a replacement.
Then the maximal size of such a replacement is some (implicit) function of $r$. % and after performing such replacements exhaustively the graph enjoys the desired structure.

This technique breaks down for weighted graphs because it is no longer true that removing the boundary of an $r$-protrusion is cheap.
Here, it is  unclear whether the number of different  ``behaviors'' can be bounded in terms of $r$ alone.

\paragraph{Sampling with weights.}
Before we deal with weighted protrusions, let us sketch how to adapt the algorithm for \fvs to the weighted setting.
First, we can still safely remove vertices of degree 1.
Even though we cannot remove all vertices of degree 2, we can preprocess the graph to ensure that each such vertex is adjacent to higher degree vertices.
After such preprocessing the number of edges is at least $\frac{5n}{4}$ and so
at least $\frac{1}{5}$-fraction of edges is incident to the optimal feedback vertex set $X$.
Observe that sampling a random endpoint of a random edge
is equivalent to sampling a vertex proportionally to its degree so $\sum_{v \in X} \deg(v) \ge \frac{1}{10} \cdot \sum_{v \in V(G)} \deg(v)$.
It is  futile though to sample vertices according to the same distribution as before because the solution may consist of light-weighted vertices only and guessing just a few heavy vertices can blow up the solution weight.
As we want to make the light-weighted vertices more preferable, we will sample a vertex $v$ of weight $w(v)$ with probability proportional to $\frac{\deg(v)}{w(v)}$.
Let $W = \sum_{v \in V(G)} \frac{\deg(v)}{w(v)}$.
Then the expected weight of the chosen vertex equals
\[\ex{w(v)} = \sum_{v \in V(G)} w(v) \cdot \frac{\deg(v)}{w(v) \cdot W} = \sum_{v \in V(G)} \frac{\deg(v)}{W}.\]
On the other hand, the expected value of the {\bf revealed part of the solution} equals 
\[\ex{w(v) \cdot 1_{v\in X}} = \sum_{v \in X} w(v) \cdot \frac{\deg(v)}{w(v) \cdot W} = \sum_{v \in X} \frac{\deg(v)}{W} \ge \frac{1}{10} \cdot \ex{w(v)}.\]
Even though the expected number of steps before we hit a vertex from $X$ might be large, the expected decrease of the optimum in each step is within a constant factor from the expected cost paid by the solution in this step.
Such an algorithm can be still proven to yield 10-approximation via a martingale argument.

\paragraph{Sampling with protrusions.}
The main obstacle remains how to deal with a situation when only a small fraction of edges is incident to a solution because of large protrusions.
Let $X$ be a treewidth-$\eta$ modulator in $G$ of minimum weight and $r = \Oh(\eta)$.
We exploit the structure of $G-X$ in a different way: since $G-X$ can be covered by few $r$-protrusions, {\bf we will sample a set $A$ from a family of $r$-protrusions that are in some sense maximal and
hope that we pick a set $A$ that intersects $X$.}
This is exactly the idea formalized in \cref{thm:main:cover}. %, which constitutes our main technical contribution.
Next, we try to guess the set $X \cap A$ or some subset $Y \sub A$ which is ``at least as good''. 
Even though the number of different ``behaviors'' of a weighted $r$-protrusion may be a priori unbounded, {\bf we can still use Courcelle's Theorem %\cite{courcelle2012graph} 
to bound the number of different ``types'' of the partial solution $X \cap A$}. % because the interface between $A$ and the rest of the graph has constant size.
%Specifically, there are explicit bounds on the sizes of representatives for each such ``type'' of a partial solution~\cite{BasteST23hittingIV}.
%Notably, these bounds work for every family $\fcal$ even when it contains disconnected graphs.
Intuitively speaking, it is easier to enumerate ``types of a partial solution'' than ``types of a protrusion'' because in the latter case we need to understand all the behaviors of subgraphs attainable after removal of an unknown vertex~set.%\micr{refs}

Suppose now that we are given an $r$-protrusion $A$ that is promised to intersect $X$.
Using the ``graph gluing'' technique~\cite{JansenK021, kim2015linear}, we enumerate a family $\acal \sub 2^A$ of size $\Oh_r(1)$ %containing non-empty subsets of $A$ 
such that either $X \cap A \in \acal$ or there exists $Y \in \acal$ for which $(X \sm A) \cup Y$ is also an optimal solution.
Next, we sample a random non-empty set from $\acal$ in such a way that the probability of picking a particular $Y \in \acal$ is proportional to $\frac{1}{w(Y)}$.
By the same calculation as before, the expected value of the revealed part of a solution is within a multiplicative factor $|\acal| = \Oh_r(1)$ from $\ex{w(Y)}$, which is the key observation to design a constant-factor approximation.
It is crucial that we can assume $X \cap A \ne \emptyset$ because the proposed sampling procedure makes sense only for sets $Y$ with a positive weight.

 \begin{figure} 
\centering
\includegraphics[scale=0.95]{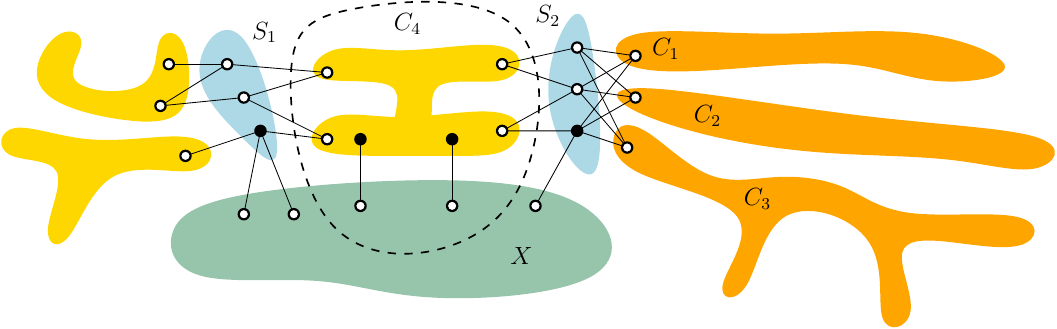}
\caption{
The set $X$ is a treewidth-$\eta$ modulator and $N(X)$ is represented by black discs.
The blue sets $S_1, S_2$ are two bags
in a tree decomposition of $G-X$ that contain the marked vertices.
The yellow and orange sets are the connected components of $G - (X \cup S_1 \cup S_2)$.
For $i = 1,2,3$ we have that $G[C_i]$ is connected, $\tw(G[C_i]) \le \eta$, and $N(C_i) = S_2$,
hence $(C_i,S_2)$ is a simple $\eta$-separation.
The pair $(C_1 \cup C_2 \cup C_3,\, S_2)$ forms a semi-simple $\eta$-separation.
The vertex set $C_4$ is enclosed by the dashed cycle.
The simple $\eta$-separation $(C_4, N(C_4))$ 
covers two edges outgoing from $X$ and so the set $N(C_4) \sm X$ has been marked.
In the proof we take advantage of the fact that there can be no inclusion-wise maximal simple $\eta$-separation $(C,S)$ for which $C$ is contained in the yellow subset of $C_4$.
} \label{fig:outline:cover}
\end{figure}

\paragraph{Main proof sketch.}
Finally, we provide some intuition behind the proof of \Cref{thm:main:cover}.
It is more convenient for us not to work directly with protrusions but with pairs $(C,S)$ of vertex sets, where $N(C) \sub S$ and $\tw(G[C]) + |S|$ is bounded.
For technical reasons, we define an {\em $\eta$-separation} as such a pair in which $\tw(G[C]) \le \eta$ and $|S| \le 2\eta + 2$.
Note that then $C \cup S$ forms a $(3\eta+2)$-protrusion.
We call an $\eta$-separation {\em simple} when $G[C]$ is connected and $S = N(C)$.

We enumerate the family of inclusion-wise maximal simple $\eta$-separations.
Since a bounded-treewidth subgraph may contain a large subfamily of such pairs $(C_i,S_i)$ for which $S_i$ coincide (see \Cref{fig:outline:cover}), we need to merge such a subfamily into a single $\eta$-separation, that is called {\em semi-simple}.
These structures
can be compared to {\em important clusters} from the work~\cite{MarxR14}.
Next, we would like every edge in the graph to be covered by some  $\eta$-separation in the family, so if some edge $uv$ has not been covered so far, we insert the pair $(\emptyset, \{u,v\})$ to the family.
Note that this forms a valid $\eta$-separation for any $\eta \ge 1$.
Let $\scal$ denote the family constructed according to this procedure.

Consider 
an unknown treewidth-$\eta$ modulator $X$ is $G$ and let $\scal^-_X \sub \scal$ denote the subset of those pairs $(C,S)$ for which $X \cap (C \cup S) =\emptyset$.
Our goal is to show that $|\scal^-_X| \le \alpha \cdot |\scal|$ for some $\alpha(\eta) < 1$.
%We want to show that the fraction of pairs $(C,S) \in \scal$ for which $X \cap (C \cup S) =\emptyset$ is bounded by some constant $\alpha < 1$.
First, for simplicity, assume that for every edge $uv$ outgoing from $X$ the pair $(\emptyset, \{u,v\})$ belongs to~$\scal$.
Then it suffices to prove that the size of $\scal^-_X$ %that fit into $G-X$ 
is bounded by $\Oh(|N(X)|)$.
This can be done using known techniques by analyzing a tree decomposition of $G-X$.
However, the assumption above is clearly unrealistic because we do not know the set $E(X,V(G)\sm X)$ and in general we cannot insert into $\scal$ pairs $(\emptyset, \{u,v\})$ covered by larger $\eta$-separations as we want to keep the number of such pairs in $\scal^-_X$ bounded. 

The problematic case occurs when there is some $(C,S) \in \scal$ that covers many edges outgoing from $X$ and so the comparison to $|N(X)|$ is no longer meaningful.
Our strategy is to {\bf mark the vertices of $G-X$ that appear as $S$-vertices 
in some $(C,S) \in \scal$ covering an edge $e \in E(X,V(G)\sm X)$.}
We show that in fact the size of $\scal^-_X$ %that have an~empty intersection with $X$
is proportional to the number of marked vertices, which is %proportional to the number of $\eta$-separation in
bounded by 
$\Oh(|\scal \sm \scal^-_X|)$.
The proof is based on a case analysis inspecting the relative location of $(C,S) \in \scal^-_X$ with respect to the marked set.
We employ a win-win argument: % which says that
in the problematic case the existence of a large simple $\eta$-separation covering many edges in $E(X,V(G)\sm X)$ implies that a large piece of $G-X$ cannot contain simple $\eta$-separations that are inclusion-wise maximal (see \Cref{fig:outline:cover}).
These arguments allow us to bound the number of pairs $(C,S) \in \scal^-_X$ %within $G-X$ 
for which $C$ is disjoint from the marked set.
To handle the remaining pairs, we make a connection between simple $\eta$-separations and important separators.
A~key observation is that when
$(C,S)$ is an inclusion-wise maximal simple  $\eta$-separation with  $X \cap (C \cup S) =\emptyset$ and $v \in C$ then $S$ forms an important $(v,X)$-separator.
The number of such separators of size $\le 2\eta + 2$ is bounded by a function of $\eta$ for each marked vertex $v \in V(G-X)$ which results in the claimed bound.

\section{Preliminaries}
\label{sec:prelim}

\paragraph{Graph notation.}

We consider simple undirected graphs without self-loops.
We follow the standard graph terminology from Diestel's book~\cite{diestel} and here we only list the most important notation.
For a vertex subset $A \sub V(G)$ we denote by $N_G(A)$ the {\em open neighborhood} of $A$, that is, the set of vertices from $V(G) \sm A$ that have a neighbor in $A$.
The {\em closed neighborhoo}d $N_G[A]$ is defined as $A \cup N_G(A)$.
{We sometimes omit the subscript when $G$ is clear from context.}
The {\em boundary} of $A$ is the set $\partial_G(A) = N_G(V(G) \sm A)$.
For a vertex $v \in V(G)$ the {\em reachability set} $R_G(v)$ of $v$ is the set of all vertices lying in the same connected component of $G$ as $v$.
We use shorthand $G-A$ for the graph $G[V(G) \setminus A]$.
A graph $G$ is {\em node-weighted} if it is equipped with a function $w \colon V(G) \to \rr^+$.
%For $v \in V(G)$, we write $G-v$ instead of $G-\{v\}$.
%For $e \in E(G)$ we denote by $G \setminus e$ a graph obtained from $G$ by removing the edge $e$.
%{For~$A \subseteq E(G)$ we denote by~$G \setminus A$ the graph with vertex set~$V(G)$ and edge set~$E(G) \setminus A$.

\paragraph*{Important separators.}
We only define separators in an asymmetric special case, sufficient for our needs.
Let $X \sub V(G)$ and $v \in V(G) \sm X$.
%Let $G$ be a graph and let $X,Y \subseteq V(G)$ be two {disjoint} sets of vertices. 
A vertex set $S \subseteq V(G) \setminus (X \cup \{v\})$ is an {\em $(v,X)$-separator} if $R_{G-S}(v) \cap X = \emptyset$.
An $(v,X)$-separator $S$ is {\em inclusion-minimal} if no proper subset of $S$ is a $(v,X)$-separator.
%and let $R$ be the set of vertices reachable from $X$ in $G - S$.
{We say that $S$ is an \emph{important $(v,X)$-separator} if it is inclusion-minimal and there is no $(v,X)$-separator $S' \subseteq V(G) \setminus (X \cup \{v\})$ such that $|S'| \leq |S|$ and $R_{G-S}(v) \subsetneq R_{G-S'}(v)$.}

\begin{lemma}[{\cite[Lemma 8.51]{cygan2015parameterized}}]\label{lem:prelim:important}
For any $X \sub V(G)$, $v \in  V(G) \sm X$, and $k \in \nn$, there are at most $4^k$
 important $(v,X)$-separators of size at most $k$. % is bounded by $4^k$.
\end{lemma}

\paragraph*{Treewidth.} %and the LCA closure.}
A tree decomposition of a graph $G$ is a pair $(T, \chi)$ where~$T$ is a rooted tree, and~$\chi \colon V(T) \to 2^{V(G)}$, such that:
\begin{enumerate}
    \item For each~$v \in V(G)$ the nodes~$\{t \mid v \in \chi(t)\}$ form a {non-empty} connected subtree of~$T$. 
    \item For each edge~$uv \in E(G)$ there is a node~$t \in V(G)$ with~$\{u,v\} \subseteq \chi(t)$.
\end{enumerate}

The \emph{width} of a tree decomposition is defined as~$\max_{t \in V(T)} |\chi(t)| - 1$. The \emph{treewidth} of a graph~$G$, denoted $\tw(G)$, is the minimum width of a tree decomposition of~$G$.

\begin{lemma}[{\cite[\S 7]{cygan2015parameterized}}]\label{lem:prelim:sparse}
    For every graph $G$ it holds that $|E(G)| \le \tw(G) \cdot |V(G)|$.
\end{lemma}

A set  $X \sub V(G)$ is a {\em treewidth-$\eta$ modulator} in $G$ if $\tw(G-X) \le \eta$.
We say that a vertex set $A \sub V(G)$ forms an {\em $\eta$-protrusion} in $G$ if $\tw(G[A]) \le \eta$ and $|\partial_G(A)| \le \eta$.

\begin{definition}\label{def:prelim:cover}
    Let $\eta,r,c \in \nn$ and $G$ be a graph.
    We say that a family $\pcal \sub 2^{V(G)}$ is a  {\em $(\eta,r,c)$\cover} for $G$ if $(i)$ $\pcal$ is non-empty, $(ii)$ each $A \in \pcal$ is an $r$-protrusion, and $(iii)$ for every treewidth-$\eta$ modulator $X \sub V(G)$ the number of sets $A \in \pcal$ for which $A\cap X \ne \emptyset$ is at least $|\pcal|\, /\, c$.
\end{definition}

\paragraph*{LCA closure.}
We will use the following concept in the  analysis of tree decompositions.

\begin{definition}\label{def:treewidth:lca}
Let $T$ be a rooted tree and $S \subseteq V(T)$. % be a set of vertices in $T$. 
We define the least common ancestor of (not necessarily distinct) $u, v \in V(T)$, denoted as {$\mathsf{LCA}(u, v)$}, to be the deepest node $x$ which is an ancestor of both $u$ and $v$.
The LCA closure of $S$ is defined as
\[
\overline{\mathsf{LCA}}(S) = \{\mathsf{LCA}(u, v): u, v \in S\}.
\]
\end{definition}

\begin{lemma}[{\cite[Lemma 1]{FominLMS12}}]\label{lem:treewidth:lca}
Let $T$ be a rooted tree, $S \subseteq V(T)$, and $L = \overline{\mathsf{LCA}}(S)$. 
Then $|L| \le 2|S|$ and for every connected component $C$ of $T-L$, $|N_T(C)| \le 2$.
\end{lemma}

\paragraph*{Minors.}
The operation of contracting an edge $uv \in E(G)$ introduces a~new vertex adjacent to all of {$N_G(\{u,v\})$}, after which $u$ and $v$ are deleted. %The result of contracting $uv \in E(G)$ is denoted $G / uv$. 
%For $A \subseteq V(G)$ such that $G[A]$ is connected, we say we contract $A$ if we simultaneously contract all edges in $G[A]$ and introduce a single new vertex.
We say that $H$ is a contraction of $G$, if we can turn $G$ into $H$ by a (possibly empty) series of edge contractions.
%We can represent the result of such a process with a mapping $\Pi \colon V(H) \to 2^{V(G)}$, such that the sets $(\Pi(h))_{h\in V(H)}$ form a partition of $V(G)$, induce connected subgraphs of $G$, 
%and $E_G(\Pi(h_1), \Pi(h_2)) \ne \emptyset$ if and only if $h_1h_2 \in E(H)$.
%This mapping is called  a contraction-model of $H$~in~$G$ and the sets $\Pi(h)$ are called branch sets.
We say that $H$ is a {minor} of $G$, denoted $H \preccurlyeq_m G$, if we can turn $G$ into $H$ by a (possibly empty) series of edge contractions, edge deletions, and vertex deletions.
The condition $H \preccurlyeq_m G$ implies $\tw(H) \le \tw(G)$.

Let $\fcal$ be a finite family  of graphs.
We say that a graph $G$ is {\em $\fcal$-minor-free} if $G$ contains no $F$ from $\fcal$ as a minor.
Whenever $\fcal$ contains a planar graph,  there exists a constant $\eta(\fcal)$ such that being $\fcal$-minor-free implies having treewidth bounded by $\eta(\fcal)$~\cite{chuzhoy2021towards}.
%It is known that $\eta(\fcal)$ is polynomial in the size of the smallest planar graph in $\fcal$
Conversely, the class of graphs of treewidth bounded by $\eta$ can be characterized by a finite family $\fcal(\eta)$ of forbidden minors that contains a planar graph~\cite{GM20}.

The condition $H \preccurlyeq_m G$ can be tested in time $\Oh_H(|V(G)|^3)$ by the classic result of Robertson and Seymour~\cite{GM13}, however there are no explicit upper bounds on the constant hidden in the $\Oh_H$-notation.
As we want to design constructive algorithms,
we will rely on the simpler algorithm by Bodlaender, which is sufficient for our applications. 

\begin{theorem}[\cite{Bodlaender96}]\label{lem:prelim:f-free}
There is a computable function $f \colon \nn \to \nn$ for which the following holds.
    For any  finite family $\fcal$ of graphs containing a planar graph,
    there is an algorithm that checks if $G$ is $\fcal$-minor-free in time $f(\max_{F \in \fcal} |V(F)|) \cdot |V(G)|$.
    Furthermore, testing if $\tw(G) \le \eta $ can be done in time $f(\eta) \cdot |V(G)|$.
\end{theorem}

\paragraph*{Boundaried graphs.} The following definitions are imported from the work~\cite{BasteST23hittingIV}.

\begin{definition}
Let $t \in  \nn$. A {\em $t$-boundaried graph} is a triple $\gbf = (G,B, \rho)$,
where $G$ is a graph, $B \subseteq  V(G)$, $|B|  = t$, and $\rho : B \rightarrow  [t]$ is a bijection, also called a {\em labeling}. 
\end{definition}

We refer to the set $B$ as the {\em boundary} of $\gbf$, also denoted $\partial(\gbf)$. 
We say that $\gbf$ is simply a {\em boundaried graph} if it is a $t$-boundaried graph for some $t \in \nn$.

We say that two boundaried
graphs $\gbf_1 = (G_1,B_1, \rho_1)$ and $\gbf_2 = (G_2,B_2, \rho_2)$ are {\em compatible} if $\rho_2^{-1} \circ  \rho_1$ is an
isomorphism from $G_1[B_1]$ to $G_2[B_2]$. Given two compatible $t$-boundaried graphs $\gbf_1 = (G_1,B_1, \rho_1)$ and $\gbf_2 = (G_2,B_2, \rho_2)$, we define $\gbf_1 \oplus  \gbf_2$ as the graph obtained if we
take the disjoint union of $G_1$ and $G_2$ and, for every $i \in  [t]$, we identify vertices
$\rho_1^{-1}(i)$ and $\rho_1^{-2}(i)$.

Given $h \in  \nn$, we 
define an equivalence relation $\equiv_h$.
We write $\gbf_1 \equiv_h \gbf_2$ if $\gbf_1$ and $\gbf_2$ are compatible and, for every graph $H$ on at most
$h$ vertices and $h$ edges and every boundaried graph $\mathbf{F}$ that is compatible with $\gbf_1$ (hence, with $\gbf_2$ as well), it holds that

\[ H \preccurlyeq_m \mathbf{F} \oplus \gbf_1 \Longleftrightarrow  H \preccurlyeq_m \mathbf{F} \oplus \gbf_2.\]

\begin{observation} \label{obs:prelim:equivalent}
    Let $\fcal$ be a family of graphs with size bounded by $h$.
    and $\mathbf{G}_1, \mathbf{G}_2, \mathbf{G}_3$ be compatible boundaried graphs. If  $\mathbf{G}_1 \equiv_h \mathbf{G}_2$ and $X \sub V(\mathbf{G}_3) \sm \partial(\mathbf{G}_3)$ is an $\fcal$-minor hitting set in $\mathbf{G}_1 \oplus \mathbf{G}_3$,
    then $X$ is an $\fcal$-minor hitting set in $\mathbf{G}_2 \oplus \mathbf{G}_3$ as well.
\end{observation}

It follows from Courcelle’s Theorem that the relation $\equiv_h$ has finitely many equivalence classes over the family of $t$-boundaried graphs, for every fixed $t \in \nn$ (see  \cite[Lemma 3.2]{bodlaender2016meta}).
A {\em set of $t$-representatives} for $\equiv_h$ is a collection containing minimum-sized representative for each equivalence class of $\equiv_h$, restricted to $t$-boundaried graphs.
Given $t,h \in \nn$ we fix a set of representatives $\rcal_h^{(t)}$ breaking the ties arbitrarily.
What is important for us, there is an explicit bound on a maximal size of a $t$-representative.

\begin{theorem}[{\cite[Theorem 6.6]{BasteST23hittingIV}}]
\label{thm:prelim:representative}
    There is computable function $f \colon \nn^2 \to \nn$ such that if $t,q,h \ge 1$ and $\mathbf{G} = (G,B,\rho)$ is a $K_q$-minor-free $t$-boundaried graph in $\rcal_h^{(t)}$, then $|V(G)| \le f(q,h)\cdot t$.
\end{theorem}

The dependence on $q$ is immaterial for our work.
Note that whenever $\mathbf{G}_1$ contains $K_h$ as a minor, then for any compatible  $\mathbf{G}_2$ and any $H$ on at most $k$ vertices it holds that $H \preccurlyeq_m \mathbf{G}_1 \oplus \mathbf{G}_2$.
Consequently, $\mathbf{G}_1 \equiv_h \mathbf{G}_3$ where $\mathbf{G}_3$ is a graph with an~isomorphic boundary plus $K_h$.
This bounds the size of the unique boundaried graph in $\rcal_h^{(t)}$ that is not $K_h$-minor-free so we can simplify the statement of \Cref{thm:prelim:representative}.
Finally, we can construct family $\rcal_h^{(t)}$ by enumerating all graphs within the given size bound.

\begin{corollary}\label{cor:prelim:representative}
     There is computable function $f \colon \nn \to \nn$ such that if $t,h \ge 1$ and $\mathbf{G} = (G,B,\rho) \in \rcal_h^{(t)}$, then $|V(G)| \le f(h)\cdot t$.
     Furthermore, family $\rcal_h^{(t)}$ can be computed in time $\Oh(2^{f(h)^2t^2})$.
\end{corollary}

\paragraph*{Minor hitting problems.}

For a finite graph family $\fcal$, a set  $X \sub V(G)$ is an {\em $\fcal$-minor hitting set} in $G$ if $G-X$ is $\fcal$-minor-free.
For a node-weighted graph $G$ %with weight function $w \colon V(G) \to \rr^+$
we denote by $\opt_\fcal(G)$ the minimum total weight of an $\fcal$-minor hitting set in $G$.
The problem of finding an $\fcal$-minor hitting set of minimum weight is called \fdel.
In the special case when $\fcal$ contains at least one planar graph, we deal with the \planardel problem.
For $k \in \nn$ we denote by $\opt_{\fcal,k}(G)$ the minimum weight of a $k$-optimal $\fcal$-minor hitting set, that is, optimal among those on at most $k$ vertices.
If there is no such set, we define $\opt_{\fcal,k}(G) = \infty$.

In the \twdel problem we seek a treewidth-$\eta$ modulator of minimum weight, denoted $\opt_\eta(G)$.
We have $\opt_\eta(G) = \opt_{\fcal(\eta)}(G)$.
In the other direction we only have inequality $\opt_{\eta(\fcal)}(G) \le \opt_\fcal(G)$ because not every graph of treewidth bounded by $\eta(\fcal)$ must be $\fcal$-minor-free.

Baste, Sau, and Thilikos~\cite{BasteST20, BasteST23hittingIV} gave an algorithm for general \ufdel on bounded-treewidth graphs.
Their dynamic programming routine stores a minimum-sized partial solution for each equivalence class in the relation $\equiv_h$ where $h$ is the maximal size of a graph in $\fcal$.
%in a certain relation over boundaried graphs (see \Cref{app:exha}).
As advocated 
in the follow-up work~\cite[\S 7]{SauST22}, this algorithm can be easily adapted to handle weights, both for computing $\opt_{\fcal}(G)$ and  $\opt_{\fcal,k}(G)$.
The first adaptation is straightforward and for the second one observe that the DP routine can store a minimum-weight partial solution for each equivalence  class and each size bound $\ell \in [k]$.
This incurs an additional factor $\poly(k) \le \poly(|V(G)|)$ in the running time.

\begin{theorem}[\cite{BasteST23hittingIV}]
\label{thm:exha:tw-algo}
    There is a computable function $f \colon \nn^2 \to \nn$ for which
    \fdel can be solved in time $f(\max_{F \in \fcal} |V(F)|,\, t) \cdot |V(G)|^{\Oh(1)}$ on graphs of treewidth at most~$t$.
    Furthermore, for a given $k \in \nn$, a $k$-optimal solution can be found within the same running time bound.
\end{theorem}%\micr{check if this reference does not work for single $F$}

%The article~\cite{BasteST23hittingIV} is concerned with the asymptotics of the function $f$ with respect to $t$ but we omit the details since in our application the value of $r$ is always constant.

%The dynamic-programming routine on a tree decomposition, 
%corollary about $k$-optimal?
%Such a modification can be easily incorporated into the algorithm from \Cref{thm:exha:tw-algo} and incurs an additional factor $\ell^{\Oh(1)}$ in the running time while we can always assume $\ell \le |V(G)|$.

\iffalse
\begin{theorem}\label{lem:prelim:tw-decomp}
    There is an algorithm that checks if $\tw(G) \le \eta$ in time $f(\eta) \cdot |V(G)|$ and, if yes, outputs a tree decomposition of $G$ of width at most $\eta$.
\end{theorem}

We state a simple corollary about testing if a graph $G$ is $\fcal$-minor-free for a family $\fcal$ containing a planar graph.
If $G$ is $\fcal$-minor-free then $\tw(G) \le \eta(\fcal)$ and this can be checked using \Cref{lem:prelim:tw-decomp}.
If the treewidth 
\fi

\paragraph{Martingales.} We will need the following concept from probability theory.

\begin{definition}
    A sequence of random variables $\mathbf{X} = (X_0,X_1,X_2,\dots)$ is called a {\em supermartingale} if for each $k \in \nn$ it holds that {\em (i)} $\ex{|X_k|} < \infty$ and  {\em (ii)}
    $\ex{X_{k+1} \mid X_0,\dots,X_k } \le X_k$.

    A random variable $\tau \in \nn$ is called a {\em stopping time} with respect to the process $\mathbf{X}$ if for each $k \in \nn$ the event $(\tau \le k)$ depends only on $X_0,\dots,X_k$.
\end{definition}

%it says that the longer we wait before we take $X_i$
The following theorem is also known in more general forms when the variable $\tau$ may be unbounded but in our applications $\tau$ will be always bounded by the number of vertices in the input graph.
Intuitively, it says that when we observe the values $X_0, X_1, \dots$ before choosing $X_\tau$, then the optimal strategy maximizing $X_\tau$ is to quit at the very beginning with $\tau=0$.

%$\mathbf{X}$ is a supermartingale and we observe the values $X_0, X_1, \dots$

\begin{theorem} [Doob's Optional-Stopping Theorem {\cite[\S 12.5]{grimmett2020probability}}]
\label{thm:prelim:doob}
Let $\mathbf{X} = (X_0,X_1,X_2,\dots)$
be a supermartingale and $\tau$ be a stopping time
with respect to $\mathbf{X}$.
Suppose that $\pr{\tau \le n} = 1$ for some integer $n \in \nn$.
If~so, then $\ex{X_{\tau}}\le\ex{X_{0}}$. 
\end{theorem}

\section{Constructing a modulator hitting family}

This section is devoted to the proof of \Cref{thm:main:cover}.
First, we need several additional definitions (to be used only within this section).
Instead of working directly with protrusions we will phrase the construction in terms of {\em $\eta$-separations}.

\begin{definition}
    For a graph $G$ and $\eta \ge 1$, an {\em $\eta$-separation} is a pair $(C,S)$ such that $C, S \sub V(G)$, $N_G(C) \sub S$, $\tw(G[C]) \le \eta$, and $|S| \le 2\eta + 2$.

    An $\eta$-separation $(C,S)$ is called {\em simple} if $G[C]$ is connected and $N_G(C) = S$.

    An $\eta$-separation $(C,S)$ is called {\em semi-simple} if for every $v \in C$ the pair $(R_{G-S}(v),S)$ forms a simple 
    $\eta$-separation.

    We define a partial order relation on $\eta$-separations $(C_1,S_1) \vartriangleleft (C_2,S_2)$ as $C_1 \sub C_2$.
    %if $C_1 \sub C_2$ and $(C_1 \cup S_1) \sub (C_2 \cup S_2)$. 
\end{definition}

We could alternatively define an $\eta$-separation with a second parameter bounding $|S|$ but it would be pointless as we never consider other bounds than $2\eta + 2$.
This bound corresponds to the maximal size of a bag in a tree decomposition of width $\eta$, times 2.
Since  $C_1 \sub C_2$ implies $N[C_1] \sub N[C_2]$ we can observe the following.

\begin{observation}\label{lem:protr:greater}
    If $(C_1,S_1)$ and  $(C_2,S_2)$ are simple $\eta$-separations and $(C_1,S_1)\vartriangleleft(C_2,S_2)$
    then $(C_1 \cup S_1) \sub (C_2 \cup S_2)$.
    %and $C_1 \subseteq C_2$ then $(C_1,S_1)\vartriangleleft(C_2,S_2)$.
\end{observation}
\iffalse
\begin{proof}
 We have $N[C_i] = C_i \cup S_i$ for $i \in \{1,2\}$ while 
     the condition $C_1 \sub C_2$ implies $N[C_1] \sub N[C_2]$.       
\end{proof}
\fi

We say that a simple $\eta$-separation  $(C_1,S_1)$ is $\vartriangleleft$-maximal if there is no other simple $\eta$-separation  $(C_2,S_2)$ for which  $(C_1,S_1) \vartriangleleft (C_2,S_2)$.

When $(C,S)$ is an  $\eta$-separation and $(T,\chi)$ is a tree decomposition of $G[C]$ of width $\le \eta$, we can easily turn it into a tree decomposition of $G[C\cup S]$ of width $3\eta+2$ by inserting the vertices of $S$ into every bag of $(T,\chi)$.
Consequently, we have the following.

\begin{observation}
    If  $(C,S)$ is an  $\eta$-separation in $G$, then $C\cup S$ is a $(3\eta+2)$-protrusion in $G$.
\end{observation}

\iffalse
    For a graph $G$ and $\eta \in \nn$, an $\eta$-separation is a pair $(C,S)$ such that $C,S \sub V(G)$ and one the following holds.
    \begin{enumerate}
        \item $\tw(C) \le \eta$, $|S| \le 2\eta$, and $N_G(C) = S$.\micr{we have $N(\emptyset) = \emptyset$}
        \item $C = \emptyset$ and $S = \{u,v\}$ for some $uv \in E(G)$.
    \end{enumerate}
\fi

\iffalse
\begin{observation}\label{lem:prelim:neigh}
    $C_1 \sub C_2$ implies $N_G[C_1] \sub N_G[C_2]$.
\end{observation}
\fi

The definition of a simple $\eta$-separation is tailored to make a convenient connection with important separators.

\begin{lemma}\label{lem:protr:important}
    Let $X \sub V(G)$ satisfy $\tw(G-X) \le \eta$, $(C,S)$ be a simple $\eta$-separation that is $\vartriangleleft$-maximal, and $v \in C$.
    If $X \cap (C \cup S) = \emptyset$ then $S$ is an important $(v,X)$-separator.
\end{lemma}
\begin{proof}
    Since $v \in C \sub V(G) \sm X$ and $N(C) = S \sub V(G) \sm X$, the set $S$ is a $(v,X)$-separator.
    We argue that $S$ is inclusion-minimal.
    Suppose that this is not the case and let $S' \subsetneq S$ be an inclusion-minimal $(v,X)$-separator.
    Let $C' = R_{G-S'}(v)$ so $N(C') = S'$.
    Observe that $C' \cap X = \emptyset$ so $\tw(G[C']) \le \eta$ and $(C', S')$ is also a simple $\eta$-separation.
    We have $C \sub C'$ and $S \sm S' \sub C'$ so $C$ is a proper subset $C'$. %\micr{maybe use just this? hmm, but edge-pairs}
    %Moreover, $C \cup S$ is a subset of $C' \cup S' = N[C']$.
    %\Cref{lem:protr:greater} implies that $(C,S) \vartriangleleft (C', S')$ and so
    Hence $(C,S)$ could not have been $\vartriangleleft$-maximal.

    Next, suppose that there is a $(v,X)$-separator $\widehat S$ for which $C = R_{G-S}(v) \subsetneq R_{G-\widehat{S}}(v)$ and $|\widehat S| \le |S|$.
    We can assume that $\widehat S$ is inclusion-minimal.
    Consider $\widehat C = R_{G-\widehat{S}}(v)$; we have $N(\widehat C) = \widehat S$.
    Again, $\widehat C \cap X = \emptyset$ so $\tw(G[\widehat C]) \le \eta$ and $(\widehat C, \widehat S)$ is a simple $\eta$-separation.
    We have $C \subsetneq\widehat C$ and $(C,S) \vartriangleleft (\widehat C,\widehat S)$.
    %We have already established th
    %We obtain  $(C,S) \vartriangleleft (\widehat C, \widehat S)$ from \Cref{lem:protr:greater} because $C \subsetneq\widehat C$.
    This yields a contradiction and implies that indeed $S$ is an important $(v,X)$-separator.
\end{proof}

%As simple $\eta$-separations behave like important separators with respect to a treewidth-$\eta$ modulator~$X$, semi-simple $\eta$-separations can be compared to {\em important clusters} from the work~\cite{MarxR14}. 

In the following construction we will use an equivalence relation $(C_1,S_1) \approx (C_2,S_2)$ defined as $S_1 = S_2$.
For a family $\{(C_1,S), \dots, (C_\ell,S)\}$ we define its {\em merge} as $(C_1 \cup \dots \cup C_\ell, S)$.
Note that when each $(C_i,S)$ is a 
simple $\eta$-separation then the merge is semi-simple.

\begin{definition}  
    We construct the family  $\shat(G,\eta)$ as follows.
    \begin{enumerate}
        \item Let $\scal_{all}(G,\eta)$ be the set of all  simple $\eta$-separations in $G$.
        \item Let $\scal_{max}(G,\eta)$ be the set of those elements in $\scal_{all}(G,\eta)$ that are $\vartriangleleft$-maximal. % with respect to relation $\vartriangleleft$.
        \item Initialize $\shat(G,\eta)$ as empty.
        \item For each $(\approx)$-equivalence class in $\scal_{max}(G,\eta)$ insert its merge into $\shat(G,\eta)$.
        \item For each $uv \in E(G)$, if there is no $(C,S) \in \scal_{all}(G,\eta)$ for which $\{u,v\} \sub C \cup S$, insert the pair $(\emptyset, \{u,v\})$ into $\shat(G,\eta)$.
    \end{enumerate}
\end{definition}

We will refer to the pairs considered in the last step as {\em edge-pairs}.
%An {\em edge-pair} $(C,S)$ is given as $C = \emptyset$ and $S = \{u,v\}$ for some $uv \in E(G)$.
Note that every edge-pair forms an  $\eta$-separation for any $\eta \ge 1$.

We make note of several important properties of this construction, to be used in the main proof.
%The main proof will be oblivious to how $\shat(G,\eta)$ is constructed and will only rely on these facts.

\begin{observation}\label{obs:protr:prop}
    The family $\shat(G,\eta)$ enjoys the following properties.
    \begin{enumerate}
        \item Each element of $\shat(G,\eta)$ is either an edge-pair or a semi-simple $\eta$-separation.\label{item:protr:semisimple}
        \item If $(C,S) \in \shat(G,\eta)$ is an edge-pair $(C,S) = (\emptyset, \{u,v\})$ then there exists no simple $\eta$-separation $(C',S')$ in $G$ for which $\{u,v\} \sub C' \cup S'$.\label{item:protr:edge-pair}
        \item If $(C,S) \in \shat(G,\eta)$ is semi-simple and $v \in C$ then $(R_{G-S}(v),S)$ is a $\vartriangleleft$-maximal simple $\eta$-separation.\label{item:protr:maximal}
        \item For each edge $uv \in E(G)$ there exists $(C,S) \in \shat(G,\eta)$ with $\{u,v\} \sub C \cup S$.\label{item:protr:edge}
        \item For each $S \sub V(G)$ there is at most one $C \sub V(G)$ for which $(C,S) \in \shat(G,\eta)$.
        \label{item:protr:unique}        
    \end{enumerate}
\end{observation}

\begin{lemma}\label{lem:protr:size}
    For an $n$-vertex graph $G$, the family $\shat(G,\eta)$ has size $n^{\Oh(\eta)}$ and it can be computed in time $n^{\Oh(\eta)}$.
\end{lemma}
\begin{proof}
    %The set $\scal_{all}(G,\eta)$ comprises all edge-pairs and simple $\eta$-separation in $G$.
    %The number of edge-pairs in $G$ is clearly at most $n^2$.
    Each simple $\eta$-separation $(C,S)$ can be characterized by the choice of $S$ and the choice of a single connected component of $G-S$.
    There are $n^{2\eta+2}$ possibilities for $S$ and $n$ possibilities to choose a component of $G-S$.
    Therefore, $|\scal_{all}(G,\eta)| \le n^{\Oh(\eta)}$.

    It suffices to observe that $|\shat(G,\eta)| \le |\scal_{max}(G,\eta)| + n^2 \le |\scal_{all}(G,\eta)| + n^2$.
    Identification of the $\vartriangleleft$-maximal elements and the  $(\approx)$-equivalence classes can be easily performed in polynomial time with respect to the size of $\scal_{all}(G,\eta)$.
\end{proof}

Before we move on to the main proof, we need one more technical lemma to bound the number of simple $\eta$-separations of a certain kind.

\begin{lemma}\label{lem:protr:lca}
    Let $(T,\chi)$ be a tree decomposition of $G$ of width $\eta$, $S \sub V(T)$, $L = \overline{\mathsf{LCA}}(S)$, and $\widehat D = \bigcup_{t\in L} \chi(t)$.
    Then there exists a family $\bcal$ of subsets of $\widehat D$ such that $|\bcal| \le 4^{\eta+2} \cdot |L| + 1$ and for each connected component $C$ of $G - \widehat D$ it holds that $N_G(C) \in \bcal$.
\end{lemma}
\begin{proof}
    First suppose that $L = \emptyset$.
    Then $\widehat D = \emptyset$ and for every connected component $C$ of $G - \widehat D = G$ we have $N_G(C) = \emptyset$.
    Hence $\bcal = \{\emptyset\}$ satisfies the lemma.

    From now on we assume $L \ne \emptyset$ and so we have  $\widehat D \ne \emptyset$.
    Consider a graph $T'$ obtained from $T$ by contracting each connected component of $T - L$ into a one of its neighbors in $L$.
    Clearly $T'$ is a tree as a contraction of $T$.
    We construct $\bcal$ as follows:
    for each $x \in L$ take all the subsets of $\chi(x)$  and for each edge $xy \in E(T')$ add take all the subsets of $\chi(x) \cup \chi(y)$.
    Since $|E(T')| = |L| - 1$ we have $|\bcal| \le 2^{\eta + 1} \cdot |L| + 4^{\eta + 1}  \cdot (|L|-1) \le 4^{\eta + 2} \cdot |L|$.
    
     \Cref{lem:treewidth:lca} states that every connected component of $T - L$ has at most 2 neighbors in $L$.
    %Consequently, every 
    %Observe that when $N(C') = \{x,y\}$ then $xy \in E(T')$.
    Consider a connected component $C$ of $G - \widehat D$.
    There is a unique connected component $C_T$ of $T - L$ such that $C \sub \bigcup_{t \in C_T} \chi(t)$.
    If $N_T(C_T) = \{x\}$ then $N_G(C)$ is contained in $\chi(x)$.
    If $N_T(C_T) = \{x,y\}$ then $xy \in E(T')$ and $N_G(C) \sub \chi(x) \cup \chi(y)$.
    In both cases $N_G(C)$ belongs to $\bcal$.
\end{proof}

We are ready to prove our main technical result, leading to \Cref{thm:main:cover}.
We need to assume $\tw(G) > \eta$ because otherwise $X$ might be empty while $\shat(G,\eta)$ would contain a single semi-simple $\eta$-separation $(V(G),\emptyset)$.
%Similarly, in \Cref{lem:protr:lca} we had to assume $S \ne \emptyset$ as otherwise each  connected component has  an empty neighborhood and $|\bcal| = 1$, $|L| = 0$.

\begin{lemma}\label{lem:protr:cover}
    Suppose that $\tw(G) > \eta$ and $\tw(G-X) \le \eta$. 
    Let $\shat(G,\eta)_X^+ \sub \shat(G,\eta)$ comprise those pairs $(C,S)$ for which $X \cap (C \cup S) \ne \emptyset$. %that intersect $X$.
    Then $|\shat(G,\eta)| \le 2^{\Oh(\eta)} \cdot |\shat(G,\eta)_X^+|$.
\end{lemma}
\begin{proof}
    First, observe that it suffices to prove the lemma for $X$ being an inclusion-minimal treewidth-$\eta$ modulator because adding vertices to $X$ can only increase $|\shat(G,\eta)_X^+|$.
    Since $\tw(G) > \eta$ we know that $X$ is non-empty.
    We also know that $N_G(X) \ne \emptyset$: otherwise removing a single vertex $u$ from $X$ would create a singleton connected component in $G - (X \sm \{u\})$ which does not increase the treewidth (i.e., $\tw(G - (X \sm \{u\})) = \tw(G-X)$) and that would contradict minimality of $X$. 
    
    We begin by constructing a certain subset $\shat_X \sub \shat(G,\eta)_X^+$.
    For each $v \in N_G(X)$ pick some $u \in X$ such that $uv \in E(G)$; then choose some $(C_v,S_v) \in \shat(G,\eta)$ such that $\{u,v\} \sub C_v \cup S_v$ and insert  $(C_v,S_v)$ to $\shat_X$.
    We know by \Cref{obs:protr:prop}(\ref{item:protr:edge}) that such a pair exists in $\shat(G,\eta)$.
    Furthermore, $u \in X$ guarantees that $(C_v,S_v) \in \shat(G,\eta)_X^+$.
    Note that we may use the same pair $(C_v,S_v)$ multiple times so it might be that $|\shat_X| < |N_G(X)|$.
    But $N_G(X) \ne \emptyset$ ensures that $\shat_X \ne \emptyset$.

    For each $(C,S) \in |\shat_X|$ we mark the vertices in $S \sm X$.
    Let $D$ denote the set of marked vertices; we have $|D| \le 2(\eta+1) \cdot |\shat_X|$.
    Next, consider a tree decomposition $(T,\chi)$ of $G-X$ of width $\le \eta$.
    For each $v \in D$ we mark some $t \in V(T)$ for which $v \in \chi(t)$.
    Let $L$ denote the LCA closure of the set of marked nodes in $T$; we have $|L| \le 2 \cdot |D|$ (\Cref{lem:treewidth:lca}).
    Finally, let $\widehat D$ denote the union of vertices appearing in $\{\chi(t) \mid t \in L\}$.
    We have $D \sub \widehat D$ and  $|\widehat D| \le (\eta + 1)\cdot |L| \le 4 (\eta + 1)^2\cdot |\shat_X|$.

    We denote $\shat(G,\eta)_X^- = \shat(G,\eta) \sm \shat(G,\eta)_X^+$, that is, this set comprises those pairs $(C,S)$ for which $X \cap (C \cup S) = \emptyset$.
Our current goal is to prove that $|\shat(G,\eta)_X^-| \le 2^{\Oh(\eta)} \cdot |\shat_X|$.  
    Since $\shat_X$ is a subset of $\shat(G,\eta)_X^+$, this will entail the claimed inequality.
    To this end, we will inspect the types of elements appearing in $\shat(G,\eta)_X^-$: they are either edge-pairs or semi-simple $\eta$-separations.
    See \Cref{fig:protr:cover} for an illustration.
    %When it comes to semi-simple $\eta$-separations $(C,S)$, we only need to bound the number of possible sets $S$ because they are distinct 
    %(\Cref{item:protr:unique}).

    \begin{figure} 
\centering
\includegraphics[scale=1.05]{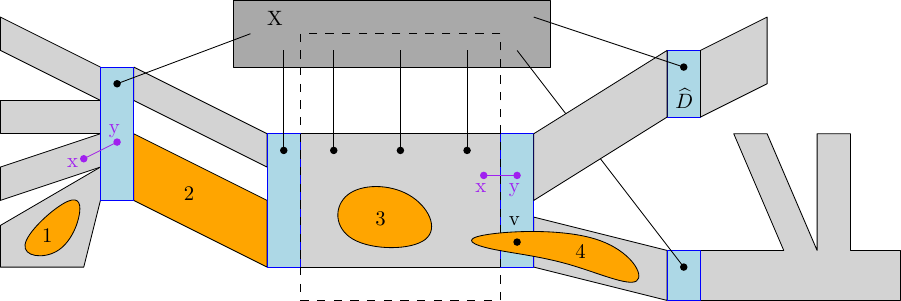}
\caption{An illustration of the case analysis in \Cref{lem:protr:cover}. 
The graph $G-X$ has treewidth bounded by $\eta$ and $\widehat{D} \sub V(G) \sm X$ is the union of vertices from the marked bags, drawn in blue.
The light-gray polygons are the connected components of $G - (X \cup \widehat{D})$.
Each orange set represents $C_{sim}$ in some scenario.
The options 1 and 2 illustrate the Case (1.a.i).
If $C_{sim}$ is a proper subset of $\widehat C$ (the gray one in option 1) then $(C_{sim},S)$ is not $\vartriangleleft$-maximal.
If $C_{sim} = \widehat C$ then $S$ is located within two bags of the tree decomposition of $G-X$ and this option boils down to \Cref{lem:protr:lca}.
The option 3 corresponds to Case (1.a.ii). Here the dashed rectangle depicts the interior of a $\vartriangleleft$-greater simple $\eta$-separation (denoted $(C^v_{sim},S_v)$ in the proof) which covers the three vertices from $N(X)$ as well as $C_{sim}$.
Also, two arrangements of the edge $xy$ from Case (2.a) are showed in purple.
The option 4 illustrates Case (1.b) with $v \in C_{sim} \cap \widehat D$. By \Cref{lem:protr:important} the set $S = N(C_{sim})$ must be an important $(v,X)$-separator.
} \label{fig:protr:cover}
\end{figure}

    \begin{enumerate}

        \item We bound the number of semi-simple $\eta$-separations in $\shat(G,\eta)_X^-$.
        Let $\bcal_X^-$ be the family of subsets $S \sub V(G) \sm X$ for which there exists $C_{sim} \sub V(G) \sm X$ such that $(C_{sim},S)$ is a $\vartriangleleft$-maximal simple $\eta$-separation in $G$.
        By \Cref{obs:protr:prop}(\ref{item:protr:maximal}) for each semi-simple $(C,S) \in \shat(G,\eta)_X^-$ there exists $C_{sim} \sub C$ so that $(C_{sim},S)$ is a $\vartriangleleft$-maximal simple $\eta$-separation.
        Hence by \Cref{obs:protr:prop}(\ref{item:protr:unique}) the number of semi-simple $\eta$-separations in $\shat(G,\eta)_X^-$ is bounded by $|\bcal_X^-|$ and it suffices to estimate the latter quantity.

        Let $S \in \bcal_X^-$ and $C_{sim} \sub V(G) \sm X$ be such that $(C_{sim},S)$ is a $\vartriangleleft$-maximal simple $\eta$-separation.

        \begin{enumerate}
            \item Suppose that $C_{sim} \cap \widehat D = \emptyset$. Let $\widehat C$ induce the connected component of $G-(X \cup \widehat D)$ that contains $C_{sim}$.
            \begin{enumerate}
                \item Suppose that $\widehat C \cap N_G(X) = \emptyset$.      
            By \Cref{lem:treewidth:lca} the set $N_{G-X}(\widehat C) = N_{G}(\widehat C)$ is contained in at most two bags from $L$ so $|N_G(\widehat C)| \le 2\eta+2$ and $(\widehat C,N_G(\widehat C))$ forms a simple $\eta$-separation.
            %Then $N(\widehat C) \sub $
            If $C_{sim}$ was a proper subset of $\widehat C$ then  %\Cref{lem:protr:greater} implies that $(C_{sim},S) \vartriangleleft (\widehat C,N(\widehat C))$ so $(C_{sim},S)$ 
            it would not be $\vartriangleleft$-maximal.
            Hence $C_{sim}$ corresponds to a connected component of $(G-X) - \widehat D$: we have $C_{sim} = \widehat C$ and $S = N_G(\widehat C)$.
            \Cref{lem:protr:lca} bounds the number of distinct neighborhoods of such components: the number of possibilities for the choice of such $S$ is bounded by $2^{\Oh(\eta)} \cdot |L| + 1 \le 2^{\Oh(\eta)} \cdot |\shat_X|$.
            Here we exploited the fact that $|L| \le 4(\eta+1)\cdot |\shat_X|$ and that $|\shat_X| \ge 1$.

            \item Now consider the case  $\widehat C \cap N_G(X) \ne \emptyset$.
            Let $v \in \widehat C \cap N_G(X)$ and $u$ be its neighbor in $X$
            that has been used during the construction of $\shat_X$.
            %By \Cref{obs:protr:prop}(\ref{item:protr:edge}) 
            Then there exists $(C_v,S_v) \in \shat_X$ for which  $\{u,v\} \sub C_v \cup S_v$.
            If $(C_v,S_v)$ was an edge-pair we would have marked the vertex $v$ but since $v \not\in \widehat D$ then $(C_v,S_v)$ must be semi-simple. % (by \Cref{item:protr:semisimple}).
            We also know that $v \in C_v$ because $v$ is not marked.
            
            Let $C^v_{sim} = R_{G-S_v}(v)$; then
        $(C^v_{sim},S_v)$ is a simple $\eta$-separation because  $(C_v,S_v)$ is semi-simple.
            We have $\widehat C \cap S_v = \emptyset$ (because $S_v \sm X \sub \widehat D$ by the marking scheme) so $C^v_{sim} \supseteq \widehat C\supseteq C_{sim}$.
            See \Cref{fig:protr:cover} for a visualization.
            In addition, $N_G[C^v_{sim}] \cap X \ne \emptyset$ implies that $C^v_{sim} \ne C_{sim}$ and so $(C_{sim},S)$ is not $\vartriangleleft$-maximal; a contradiction.
            %We infer from \Cref{lem:protr:greater} that $(C_{sim},S) \vartriangleleft (C^v_{sim},S_v)$ what contradicts the maximality of $(C_{sim},S)$.
            \end{enumerate}
     
        \item Suppose that $C_{sim} \cap \widehat D \ne \emptyset$.      
        %Let $(C,S) \in \shat(G,\eta)_X^-$ be semi-simple and $C \cap \widehat D \ne \emptyset$.
        Let $v \in C_{sim} \cap \widehat D$. % and $C' = R_{G-S}(v)$; then $(C',S)$ is simple and $\vartriangleleft$-maximal.\micr{ref}
        From \Cref{lem:protr:important} we know that $S$ must be an important $(v,X)$-separator.
        For each  $v \in \widehat D$ the number of important $(v,X)$-separators of size $\le(2\eta+2)$ is bounded by $4^{2\eta+2} = 2^{\Oh(\eta)}$ (\Cref{lem:prelim:important}).
        Therefore the maximal number of sets $S$ occurring in this scenario is $2^{\Oh(\eta)} \cdot |\widehat D| \le 2^{\Oh(\eta)} \cdot |\shat_X|$.

        \end{enumerate}

        \item We bound the number of edge-pairs in $\shat(G,\eta)_X^-$.    
        When $(C,S) \in \shat(G,\eta)_X^-$ is an edge-pair then
        $S = \{x,y\}$ for some $xy \in E(G-X)$ and $C = \emptyset$.
        By \Cref{obs:protr:prop}(\ref{item:protr:edge-pair}) we know that
        there can be no simple $\eta$-separation $(C',S')$ in $G$ for which $\{x,y\} \sub C' \cup S'$.
        %that is  $\vartriangleleft$-greater than $(C,S) = (\emptyset,\{x,y\})$.
        \begin{enumerate}
            \item Suppose that $\{x,y\} \not\sub \widehat D$. W.l.o.g.  assume that $x \not\in \widehat D$.
            Let $\widehat C$ % = R_{G-(X\cup\widehat D)}(x)$
            induce the connected component of $G-(X\cup\widehat D)$ that contains $x$. 
            \begin{enumerate}
            
            \item If $\widehat C \cap N_G(X) = \emptyset$,  we can use an analogous argument as in Case (1.a.i): 
            $\{x,y\} \sub N[\widehat C]$ so $(\widehat C, N_G(\widehat C))$ is a simple $\eta$-separation covering $xy$; a contradiction.
            %$(\emptyset,\{x,y\}) \vartriangleleft (\widehat C, N_G(\widehat C))$, what is impossible due to \Cref{obs:protr:prop}(\ref{item:protr:edge-pair}).
            %contradicts the maximality of $(\emptyset,\{x,y\})$.
            %Note that $y \in N_G[\widehat C]$ and the two pairs above must be different because $\widehat C \ne \emptyset$.
           %So this case cannot happen.
            
            \item If $\widehat C \cap N_G(X) \ne \emptyset$, we follow the argument from Case (1.a.ii): we use an~arbitrary $v \in \widehat C \cap N_G(X)$ to certify the existence of a simple $\eta$-separation $(C^v_{sim},S_v)$ for which $\widehat C \sub C^v_{sim}$.
            Then again $\{x,y\} \sub N_G[\widehat C] \sub N_G[C^v_{sim}]$ yields a contradiction.
            \end{enumerate}

            \item In remains to handle the case $\{x,y\} \sub \widehat D$.
            Then $xy$ is an edge in the graph $G[\widehat D]$.
            But the treewidth of this graph is at most $\eta$ so the number of edges in $G[\widehat D]$ is bounded by $\eta \cdot |\widehat D| \le 4(\eta+1)^3 \cdot |\shat_X|$ (\Cref{lem:prelim:sparse}).
            
        \end{enumerate}
     \end{enumerate}
     
We have established that $|\shat(G,\eta)_X^-| \le 2^{\Oh(\eta)} \cdot |\shat_X|$. 
    Recall that $\shat_X \sub \shat(G,\eta)_X^+$.
    This implies that $|\shat(G,\eta)| = |\shat(G,\eta)_X^+| + |\shat(G,\eta)_X^-| \le |\shat(G,\eta)_X^+| + 2^{\Oh(\eta)} \cdot |\shat_X| \le 2^{\Oh(\eta)} \cdot |\shat(G,\eta)_X^+|$, as promised.     
\end{proof}

To derive \Cref{thm:main:cover} from \Cref{lem:protr:cover}, recall that
 when $(C,S)$ is an $\eta$-separation then $(C\cup S)$ is a $(3\eta+2)$-protrusion.
 Hence $\shat(G,\eta)$ constitutes an $(\eta,3\eta+2, 2^{\Oh(\eta)})$\cover.
 The size bound and the running time to construct $\shat(G,\eta)$ follow from \Cref{lem:protr:size}.
    
\section{Exhaustive families}
 \label{sec:exha}

We explain how to enumerate the relevant partial solutions within an $r$-protrusion with Courcelle's Theorem.
A similar idea has been used for trimming the space of partial solutions in DP algorithms on tree decomposition~\cite{BasteST20, fomin2016efficient}.
Our treatment of protrusions is a combination of ideas from~\cite{kim2015linear} for processing protrusion decompositions and from~\cite{JansenK021} for processing hybrid graph decompositions.
We adapt their methods to the weighted setting.
%The following concept of an exhaustive family has been used in~\cite{JansenK021} to deal with unweighted problems on hybrid graph decompositions.

\begin{definition}[\cite{JansenK021}]
    Let $\fcal$ be a finite family of graphs and $A$ be a vertex subset in a node-weighted graph $G$.
    We say that a family $\acal \sub 2^A$ of  subsets of $A$ is {\em $(A,\fcal)$-exhaustive} if for every  $\fcal$-minor hitting set $X$ there exists $X_A \in \acal$ such that $X' = (X \sm A) \cup X_A$ is also an $\fcal$-minor hitting set and $w(X') \le w(X)$.

    Additionally, for $\ell \in \nn$ we say that a family $\acal \sub 2^A$ of  size-$(\le \ell)$ subsets of $A$ is {\em $(A,\fcal,\ell)$-exhaustive} if for every  $\fcal$-minor hitting set $X$ satisfying $|X \cap A|  \le \ell$ there exists $X_A \in \acal$ such that $X' = (X \sm A) \cup X_A$ is also an $\fcal$-minor hitting set and $w(X') \le w(X)$.
\end{definition}

When $X$ is an optimal solution then the replacement $X_A$ cannot be cheaper than $X \cap A$ as otherwise we would have $w(X') < w(X)$.
In particular, whenever $X$ is optimal and $X \cap A$ is non-empty then $w(X_A) = w(X \cap A) > 0$ so $X_A$ must be non-empty as well.

\begin{observation}\label{lem:exha:opt-normal}
    Let $X$ be an optimal $\fcal$-minor hitting set in $G$, $A$ be a vertex subset intersecting $X$, and $\acal \sub 2^A$ be an $(A,\fcal)$-exhaustive family.
    Then there exists a non-empty $X_A \in \acal$ such that $\opt_\fcal(G) - \opt_\fcal(G-X_A) = w(X_A)$.
\end{observation}

An analogous statement holds for $k$-optimal solutions as well but the justification is slightly less obvious so we provide a full proof.

\begin{lemma}\label{lem:exha:opt-k}
    Let $k,\ell \ge 1$, $X$ be a $k$-optimal $\fcal$-minor hitting set in $G$, $A$ be a vertex subset satisfying $|X \cap A| = \ell$, and $\acal \sub 2^A$ be an $(A,\fcal,\ell)$-exhaustive family.
    Then there exists a non-empty $X_A \in \acal$ such that $\opt_{\fcal,k}(G) - \opt_{\fcal,k-\ell}(G-X_A) = w(X_A)$.
\end{lemma}
\begin{proof}
    By the definition of an $(A,\fcal,\ell)$-exhaustive family there exists $X_A \in \acal$ such that $|X_A| \le \ell$, $X' = (X \sm A) \cup X_A$ is an $\fcal$-minor hitting set, and $w(X') \le w(X)$.
    We have $|X'| = |X \sm A| + |X_A| \le (k-\ell) + \ell = k$.
    But $X$ is $k$-optimal so in fact we must have $w(X') = w(X)$ and $w(X_A) = w(X \cap A)$.
    In particular, this implies that $X_A \ne \emptyset$ because $w(X \cap A) > 0$ by  assumption.
    Next, $X \sm A$ forms an $\fcal$-minor hitting set in $G-X_A$ of size $\le k-\ell$ so  $\opt_{\fcal,k-\ell}(G-X_A) \le w(X \sm A) = w(X) - w(X \cap A) = \opt_{\fcal,k}(G) - w(X_A)$.
    On the other hand, we have $\opt_{\fcal,k}(G) \le \opt_{\fcal,k-\ell}(G-X_A) + w(X_A)$ so in fact these two quantities must be equal.
\end{proof}

We take advantage of the fact that when $A$ is an $r$-protrusion then the graph $G-A$ can ``communicate'' with $A$ only through $\partial(A)$.
To construct the exhaustive families, we utilize the bounds on the minimum-sized representatives in the 
the relation $\equiv_h$ over boundaried graphs~\cite{BasteST23hittingIV}.
%Courcelle's Theorem implies that there are only $\Oh_r(1)$ many non-equivalent ``types'' of the subgraph $G - (A \cup X)$: for each of them it suffices to consider a single representative and compute a corresponding optimal partial solution on $A$. 
%The proof of the lemma below relies on a bound on the sizes of such  representatives from~\cite{BasteST23hittingIV} and is similar in spirit to~\cite[Lemma 5.33]{JansenK021-arxiv} so we have placed it in~\Cref{app:exha}.
The following proof is similar in spirit to~\cite[Lemma 5.33]{JansenK021-arxiv}.
 The difference is that here we compute partial solutions using an FPT algorithm parameterized by treewidth while therein %~\cite{JansenK021-arxiv}
 they were computed with an FPT algorithm parameterized by the solution size.

\begin{restatable}{lemma}{lemExha}
\label{lem:exha:compute}
    %For each $h \in \nn$ there is a constant $c_$
    There is a computable function $f \colon \nn^2 \to \nn$  so that the following holds.
    Let $r, h \in \nn$ and
     $\fcal$ be a family of graphs with sizes bounded by $h$.
     Next, let $A \sub V(G)$ be an $r$-protrusion in a graph $G$.
    Then there exists an $(A,\fcal)$-exhaustive family $\mathcal{A} \sub 2^A$ of size at most $f(h,r)$ and it can be computed in time $f(h,r)\cdot |V(G)|^{\Oh(1)}$. 

    Furthermore, for each $\ell \in \nn$ there exists an $(A,\fcal,\ell)$-exhaustive family $\mathcal{A} \sub 2^A$ of size at most $f(h,r)$ and it can be computed in time $f(h,r)\cdot |V(G)|^{\Oh(1)}$. 
\end{restatable}
\begin{proof}
Recall that $\partial_G(A) = N_G(V(G) \sm A)$.
We will also denote $A \sm \partial_G(A)$ by $\triangle_G(A)$.
We first give an algorithm to construct a family $\acal \sub 2^A$ and then we will prove that it is $(A,\fcal)$-exhaustive.

We begin by computing families $\rcal_h^{(\ell)}$ for each $\ell \in [r]$ (\Cref{cor:prelim:representative}).
For each subset $S \sub \partial_G(A)$ consider the graph $G - S$, in which  $A \sm S$ still forms an $r$-protrusion. % in $G - S$.
Let $\ell = |\partial_G(A) \sm S|$ and consider the 
boundaried graph $\mathbf{G}_S = (G[A \sm S], \partial_{G-S}(A \sm S), \rho)$ where $\rho$ is an arbitrary labeling.
Next, for each  $\mathbf{H} \in \rcal_h^{(\ell)}$ that is compatible with $\mathbf{G}_S$, consider the graph $G_{S,H} = \mathbf{G}_S \oplus \mathbf{H}$.
We define the weight function on  $G_{S,H}$ by preserving the weights from $G$ on $\triangle_G(A)$ and setting the remaining weights to~$\infty$.
Note that the size of $\mathbf{H}$ is bounded by some function of $h$ and $r$ (\Cref{cor:prelim:representative}) so the treewidth of $G_{S,H}$ is likewise bounded.
Therefore, we can take advantage of \Cref{thm:exha:tw-algo} to solve \fdel on $G_{S,H}$ in time $f(h,r) \cdot |V(G_{S,H})|^{\Oh(1)}$; let $X_H$ denote the computed optimal solution.
If there is no finite-weight solution then we move on to the next  $\mathbf{H} \in \rcal_h^{(\ell)}$, otherwise we have $X_H \sub \triangle_G(A)$. In this case we insert $S \cup X_H$ into $\acal$.
The size of $\acal$ can be bounded by the number of choices of $S$, that is, $2^r$, times the bound on $|\rcal_h^{(\ell)}|$, for $\ell \le r$, which follows from \Cref{cor:prelim:representative}.

We argue that $\acal$ is $(A,\fcal)$-exhaustive.
Consider an $\fcal$-minor hitting set $X$. % intersecting $A$.
Let $\widehat S = X \cap \partial_G(A)$ and  $\mathbf{G}^X = (G - (X \cup \triangle_G(A)), \partial_G(A) \sm X, \rho)$; note that $\mathbf{G}^X \oplus \mathbf{G}_{\widehat S}  = G - (X \sm \triangle_G(A))$.
Next, let $\mathbf{H} \in \rcal_h^{(\ell)}$ be
a boundaried graph that is $\equiv_h$-equivalent to $\mathbf{G}^X$.
By \Cref{obs:prelim:equivalent} set $Y \subseteq \triangle_G(A)$ is an $\fcal$-minor hitting set in $\mathbf{G}^X \oplus \mathbf{G}_{\widehat S}$ if and only if it is an $\fcal$-minor hitting set in $\mathbf{H} \oplus \mathbf{G}_{\widehat S}$.
Consequently, $X \cap \triangle_G(A)$ is an $\fcal$-minor hitting set in $\mathbf{H} \oplus \mathbf{G}_{\widehat S}$.
This means that the algorithm above must have detected some $X_H \sub \triangle_G(A)$ of weight $w(X_H) \le w(X \cap \triangle_G(A))$ and inserted $\widehat{S} \cup X_H$  into $\acal$. 
On the other hand, $X_H$ forms an $\fcal$-minor hitting set in $\mathbf{G}^X \oplus \mathbf{G}_{\widehat S} = G - (X \sm \triangle_G(A))$.
So $X' = (X \sm \triangle_G(A)) \cup X_H$ is an $\fcal$-minor hitting set in $G$ and $w(X') \le w(X)$.
As $X \cap \partial_G(A) = \widehat{S}$ we can alternatively write $X' = (X \sm A) \cup (\widehat{S} \cup X_H)$.
%Because $X$ is inclusion-minimal, $X \sm A$ is not an $\fcal$-minor hitting set
%we must have $w(X) \le w(X')$ and so $W(X \cap A) \le w(\widehat{S} \cup X_H)$.
%This means that in fact $W(X \cap A) = w(X_H)$, $w(X') = w(\widehat{S} \cup X)$, and in particular 
%so $\widehat{S} \cup X_H \ne \emptyset$ and it has been inserted into $\acal$.
This means that $\acal$ includes a replacement for $X \cap A$ and proves that $\acal$ is indeed $(A,\fcal)$-exhaustive.

Constructing an $(A,\fcal,\ell)$-exhaustive family is analogous.
The only difference is that when considering the graph $G_{S,H}$ we compute an $\ell$-optimal solution, i.e., a minimum-weight $\fcal$-minor hitting set among those on at most $\ell$ vertices.
This is doable within the same time bound (\Cref{thm:exha:tw-algo}).% and incurs an additional factor $\ell^{\Oh(1)}$ in the running time while we can always assume $\ell \le |V(G)|$.
%The exchange argument this time follows from \Cref{lem:exha:opt-k}.
\end{proof}

\section{Applications}
\subsection{The approximation algorithm}

For the ease of presentation, we first give an algorithm for \twdel and then use it as a subroutine to solve general \planardel.
We begin with the analysis of a single iteration  and compare the expected drop in the optimum to the expected cost of the sampled partial solution.

\begin{lemma}\label{lem:apx:step}
    %Let $\fcal$ be a finite family of graphs containing a planar graph and $\eta$ be the treewidth bound .
    For each $\eta \in \nn$ there is a constant $c \in \nn$ and an algorithm that, given a node-weighted $n$-vertex graph $G$ %on $n$ vertices 
    with $\tw(G) > \eta$, runs in time $\Oh_\eta\br{ n^{\Oh(\eta)}}$, and outputs a non-empty random subset $Y \sub V(G)$ with a following guarantee.
    \[\ex{w(Y)} \le c \cdot \ex{\opt_\eta(G) - \opt_\eta(G-Y)}  \]
\end{lemma}
\begin{proof}
    We apply \Cref{thm:main:cover} to compute an $(\eta, r, c_1)$\cover $\pcal$ for $G$ with $r = \Oh(\eta)$ and $c_1 = 2^{\Oh(\eta)}$.
    This takes time $n^{\Oh(\eta)}$. % and guarantees  $|\pcal| = n^{\Oh\eta+1}$.
    Next, we use \Cref{lem:exha:compute} to compute, for each $A \in \pcal$, an $(A,\fcal(\eta))$-exhaustive family $\acal_A$ of size $\le c_2$, where $c_2 = f(\max_{F\in \fcal(\eta)}|V(F)|,r)$ and $f$ is the function from that lemma.
    This takes time $n^{\Oh(\eta)} \cdot \Oh_\eta\br{ n^{\Oh(1)}}$.
    Note that $c_1,c_2$ are constants depending only on $\eta$ and
    recall that a graph has treewidth bounded by $\eta$ if and only if it is $\fcal(\eta)$-minor-free so $\opt_\eta(G) = \opt_{\fcal(\eta)}(G)$.
    Now, consider the family $\bcal$ of all pairs $(A,Y)$ where $A \in \pcal$ and $Y\in \acal_A$ is non-empty.
    We sample an element $(A,Y)$ from $\bcal$ with probability proportional to $\frac{1}{w(Y)}$.
    Note that $w(Y) > 0$ because $Y \ne \emptyset$.
    Let $\widehat Y$ denote the corresponding random variable.
    
    We claim that the random variable $\widehat Y$ satisfies the guarantee of the lemma.
    Let $W = \sum_{(A,Y) \in \bcal} \frac{1}{w(Y)}$.
    Then the probability of choosing a pair $(A,Y)$ equals $\frac{1}{w(Y) \cdot W}$.
    We have
    
    \[\ex{w(\widehat Y)} = \sum_{(A,Y) \in \bcal} w(Y) \cdot \frac{1}{w(Y) \cdot W} = \frac{|\bcal|}{W}. \]

    Now we estimate the expected value of $\opt_\eta(G) - \opt_\eta(G-\widehat Y)$.
    Clearly, for every $Y$ this quantity is non-negative.
    %We will identity a sufficiently large subset of $\bcal$ 
    Let $X$ be an optimal solution to \twdel on~$G$.
    By the definition of a modulator hitting family, we know that there is a subfamily $\pcal_X \sub \pcal$ of size 
    at least $|\pcal|/c_1$ so that for each $A \in \pcal_X$ we have $A \cap X \ne \emptyset$.
    Furthermore, \Cref{lem:exha:opt-normal} implies that for each  $A \in \pcal_X$ there is some non-empty $Y_A \in \acal_A$ satisfying $\opt_\eta(G) - \opt_\eta(G-Y_A) = w(Y_A)$.
    In particular, we have $(A,Y_A) \in \bcal$
    and we can estimate 

    \[\ex{\opt_\eta(G) - \opt_\eta(G-\widehat Y)} \ge
    \sum_{A \in \pcal_X} w(Y_A) \cdot \frac{1}{w(Y_A) \cdot W} = \frac{|\pcal_X|}{W}. \]
    It holds that $|\bcal| \le c_2 \cdot|\pcal| \le c_2 \cdot c_1 \cdot|\pcal_X|$.
    Hence the claimed inequality holds for $c = c_1 \cdot c_2$.   
\end{proof}

In the proof it was important that when $Y \sub V(G)$ belongs to multiple $(A_i,\fcal(\eta))$-exhaustive families for different sets $(A_i)$ then we consider all its copies in the family $\bcal$.
Alternatively, we could multiply the probability of choosing $Y$ by the number of such sets $(A_i)$.
By the same reason we had to include the factor $\deg(v)$ in the probabilities in the simplified algorithm for \wei\fvs in \Cref{sec:outline}.

We will perform the described sampling iteratively.
At each step, the expected value of the revealed part of a solution is within a factor $c$ from the cost of the chosen set $Y$.
The algorithm is terminated when the value of optimum drops to 0, so the expected value of $w(Y_1) + w(Y_2) + \dots$ should be bounded by $c$ times the initial optimum.
We will formalize this intuition by analyzing a random process being a supermartingale (see \Cref{sec:prelim}).

\begin{lemma}\label{lem:apx:doob}
For each $\eta \in \nn$ there is a constant $c \in \nn$ and a randomized algorithm that, given a node-weighted $n$-vertex graph $G$ graph $G$, runs in time $\Oh_\eta\br{ n^{\Oh(\eta)}}$, and outputs a treewidth-$\eta$ modulator $X$ in $G$ such that $\ex{w(X)} \le c \cdot \opt_\eta(G)$.
\end{lemma}
\begin{proof}
    We will use the constant $c$ defined in \Cref{lem:apx:step}. 
    Consider the following random process $(X_i)$ where $X_i \sub V(G)$.
    We start with $X_0 = \emptyset$. Next, for $i \ge 0$ we compute a random non-empty subset $Y_{i+1} \sub V(G) \sm X_i$ by applying \Cref{lem:apx:step} to the graph $G_i = G - X_{i}$.
    Then we set $X_{i+1} = X_{i} \cup Y_{i+1}$.
    The process stops at iteration $\tau$ when $\tw(G_\tau)$ becomes bounded by $\eta$ and the algorithm returns $X = X_\tau$.
    Note that $X_i$ grows at each iteration so $\tau$ is always bounded by $|V(G)|$.
    We will show that $X$ is a $c$-approximate solution in expectation.

    We claim that the process $P_i = w(X_i) + c \cdot \opt_\eta(G_i)$ is a supermartingale, which means that
      $\ex{P_{i+1}  \mid P_0, \dots P_i} \le P_i$.
      This condition is equivalent to  $\ex{P_{i+1} - P_i \mid P_0, \dots P_i} \le 0$.

    \[ P_{i+1} - P_i = w(Y_{i+1}) + c \cdot \br{ \opt_\eta(G_{i+1}) -  \opt_\eta(G_{i})} = w(Y_{i+1}) - c \cdot \br{ \opt_\eta(G_i) -  \opt_\eta(G_i - Y_{i+1})}\]

    \Cref{lem:apx:step} states that the expected value of such a random variable is upper bounded by 0 for any fixed graph $G_i$. %, i.e., $\ex{P_{i+1} - P_i \mid G_i} \le 0$.
    Hence the process $(P_i)$ is indeed a supermartingale.
    Consider the stopping time~$\tau$, for which it holds that $\opt_\eta(G_\tau) = 0$.
    We can now take advantage of the Doob's Optional-Stopping Theorem (\ref{thm:prelim:doob}) which states that $\ex{P_\tau} \le \ex{P_0}$.
    %We consider the stopping time $\tau$, for which we have $\opt_\eta(G_\tau) = 0$, and calculate

    \[\ex{w(X)} = \ex{w(X_\tau) + c \cdot \opt_\eta(G_\tau)} = \ex{P_\tau} \le \ex{P_0} =  w(X_0) + c \cdot \opt_\eta(G_0) = c\cdot \opt_\eta(G)\]
    This concludes the proof.
\end{proof}

Finally, we solve \planardel by first reducing the treewidth to $\eta(\fcal)$.

\thmApx*
\begin{proof}
    Any $\fcal$-minor hitting set is also a treewidth-$\eta(\fcal)$ modulator so
    %Let $\eta = \eta(\fcal)$.\micr{define}
    %for any graph $G$     
    it holds that $\opt_{\eta(\fcal)}(G) \le \opt_\fcal(G)$.
    We apply \Cref{lem:apx:doob} to compute a set $X_1 \sub V(G)$ such that $\tw(G-X_1) \le \eta(\fcal)$ and  $\ex{w(X_1)} \le c \cdot \opt_{\eta(\fcal)}(G)$.
    Since  the treewidth of $G-X_1$ is bounded, we can now solve \fdel optimally on $G-X_1$ using \Cref{thm:exha:tw-algo}.
    We compute $X_2 \sub V(G) \sm X_1$ such that $w(X_2) = \opt_\fcal(G-X_1) \le \opt_\fcal(G)$.
    The union $X_1 \cup X_2$ forms a solution that is $(c+1)$-approximate in expectation.
\end{proof}

\subsection{The single-exponential FPT algorithm}

Unlike for approximation, this time we cannot afford to first solve \twdel and then process a bounded-treewidth graph, so we will treat the bounded-treewidth scenario as a special case in the branching algorithm.
The main difference compared to \Cref{thm:main:apx} is that now we will sample both the protrusion $A \in \pcal$ and the replacement $Y \in \acal$ according to a uniform distribution.
We also need to guess the integer $\ell = |X \cap A|$ to guide the construction 
of an $(A,\fcal,\ell)$-exhaustive family;
note that we can compute such a family of size $\Oh_\fcal(1)$ only when $\ell$ is known.
However, we cannot sample $\ell$ uniformly from $[1,k]$ because when the correct value of $\ell$ is always $\Oh(1)$ we need to achieve constant success probability in every step.
As a remedy, we will use geometric distribution to sample $\ell$.

\thmFpt*
\begin{proof}
    We will abbreviate $\eta = \eta(\fcal)$.
    Consider an instance $(G, k)$ and proceed as follows.
    If $k=0$ check if $G$
    is $\fcal$-minor-free in time $\Oh_\fcal(n)$ (\Cref{lem:prelim:f-free});
    if yes return an empty set and otherwise abort the computation.
    If $\tw(G) \le \eta$ solve the problem optimally in time $\Oh_\fcal(n^{\Oh(1)})$ using \Cref{thm:exha:tw-algo}. % (see the proof of \Cref{lem:exha:compute} for a remark about adapting that algorithm to detect $k$-optimal solutions).\micr{ref to prelims}
    
    We will assume now that $k \ge 1$ and $\tw(G) > \eta$.
    In order to prepare a recursive call, we make three random choices.
    \begin{enumerate}
        \item  Use \Cref{thm:main:cover}
    to compute an $(\eta, r, c_1)$\cover $\pcal$ % of size $\le n^{\eta+1}$ 
    for $r = \Oh(\eta)$ and $c_1 = 2^{\Oh(\eta)}$.
    This takes time $n^{\Oh(\eta)}$.
    Then choose a set $A \in \pcal$ uniformly at random.
    
    \item Consider a random variable $L \in [1,k]$ where the probability of choosing $L = \ell$ is proportional to $2^{-\ell}$.
    Note that we have $\pr{L = \ell} \ge 2^{-\ell}$ for each $\ell \in [1,k]$.
 
 \item Use \Cref{lem:exha:compute} to compute an $(A,\fcal,L)$-exhaustive family $\acal$ of size $\le c_2$, where $c_2 = f(\max_{F\in \fcal}|V(F)|, r)$ and $f$ is the function from that lemma. This takes time $\Oh_\fcal(n^{\Oh(1)})$.
Then choose a non-empty set $Y \in \acal$ again uniformly at random (if $\acal = \{\emptyset\}$ abort the computation).
   Add $Y$ to the solution and recurse on $(G-Y, k - L)$.
   \end{enumerate}

    %We claim that we probability at least $1/c$ the algorithm chooses a set $Y$\micr{define exhaustive family to mark the smallest solution}

    Let $Q(G,k)$ denote the event that 
    the algorithm returns a $k$-optimal solution to the instance $(G,k)$ under an assumption that $\opt_{\fcal,k}(G) < \infty$. %$(G,k)$ admits a solution.
    We will prove that $\pr{Q(G,k)} \ge (2c_1c_2)^{-k}$ by induction on $k$.
    The claim clearly holds for $k=0$ so we will 
    assume that $k > 0$ and that the claim holds for
    all $k' < k$.
    First, if $\tw(G) \le \eta$ then the algorithm always returns a $k$-optimal solution.
    We will now assume $\tw(G) > \eta$ and analyze the three random guesses performed before a recursive call.

    Let $X$ be a fixed $k$-optimal solution to \fdel.
    It is also a solution to \twdel. % for $\eta = \eta(\fcal)$.
    %Let $\pcal_X \sub \pcal$ be a subfamily of those $A \in \pcal_X$ for which $A \cap X \ne \emptyset$.
    By \Cref{def:prelim:cover} of an $(\eta,r,c_1)$\cover, we know that $\pr{A \cap X \ne \emptyset} \ge 1/c_1$ for a random choice of $A \in \pcal$.
    Suppose that $A$ has been chosen successfully and let $\ell_A = |A \cap X| \ge 1$.
    Then  $\pr{L = \ell_A} \ge 2^{-\ell_A}$.
    %Assume that the first two guesses were correct.
    %The probability of picking the right $\ell$ in the second step equals  $2^{-\ell}$.
     \Cref{lem:exha:opt-k} states that the $(A,\fcal,\ell_A)$-exhaustive family $\acal$ contains a non-empty set  $X_A$ of size $\le \ell_A$ for which $\opt_{\fcal,k}(G) - \opt_{\fcal,k- \ell_A}(G-X_A) = w(X_A)$. 
    The probability of choosing $Y = X_A$ in the third step is $1/|\acal| \ge 1/c_2$.
    The call $(G,k)$ will be processed correctly if all these guesses are correct and the call $(G-X_A,k-\ell_A)$ is processed correctly.
    We are ready to estimate the probability of $Q(G,k)$, by first analyzing the success probability for a fixed $A \in \pcal$.

    \begin{align*}  
    \pr{Q(G,k) \mid A \cap X \ne \emptyset} & \ge  \pr{L = \ell_A} \cdot \pr{Y = X_A} \cdot \pr{Q(G-X_A,k-\ell_A)} \\
    & \ge  2^{-\ell_A} \cdot c_2^{-1} \cdot (2c_1c_2)^{-(k-\ell_A)} \ge (2c_2)^{-1} \cdot (2c_1c_2)^{-(k-1)} 
     \end{align*}

      We have used the inductive assumption and the fact that $\ell_A \ge 1$ when $A \cap X \ne \emptyset$.
     
     \[ \pr{Q(G,k)}  \ge  \pr{A \cap X \ne \emptyset} \cdot \pr{Q(G,k) \mid A \cap X \ne \emptyset} \ge c_1^{-1} \cdot (2c_2)^{-1} \cdot (2c_1c_2)^{-(k-1)}  =(2c_1c_2)^{-k}\]
     
     We have thus proven that the success probability of the presented procedure is at least $c^{-k}$ for $c = 2c_1c_2$ whereas its running time is $\Oh_\fcal\br{n^{\Oh(\eta)}}$.
     Note that if the computation has not been aborted, then the returned solution has at most $k$ vertices.
     We repeat this procedure $c^k$ many times and output the outcome with the smallest weight.
     This gives constant probability of witnessing a $k$-optimal solution.
     %Repeating this procedure $c^k$ many times yields constant probability of witnessing a $k$-optimal solution so
     %The constant hidden in the $\Oh_\fcal$-notation can be incorporated into the constant $c$ so
     The total running time is $c^k \cdot \Oh_\fcal\br{n^{\Oh(\eta)}} =  2^{\Oh_\fcal(k)}\cdot n^{\Oh(\eta)}$.
\end{proof}

\section{Conclusion}

We have shown that constant approximation for {\sc Planar} \ufdel can be extended to the weighted setting, by introducing a new technique that circumvents the need to perform weighted protrusion replacement.
%Our single-exponential FPT algorithm yields an improvement already for unweighted graphs as it does not require all the graphs in $\fcal$ to be connected.\micr{false} %, unlike the one in~\cite{FominLMS12}.
There are several natural questions about potential refinements of \Cref{thm:main:apx}.
Can the algorithm be derandomized? 
Can we improve the running time so that the degree of the polynomial is a universal constant that does not depend on $\eta$? %, as in~\cite{FominLMS12, GuptaLLM019}?
Can the approximation factor be bounded by a slowly growing function of $\eta$? %, as in \cite{GuptaLLM019}?
All these improvements are possible in the unweighted realm~\cite{FominLMS12, GuptaLLM019}.
It is however unknown whether there exists a universal constant $c$ for which every unweighted {\sc Planar} \ufdel problem admits a $c$-approximation algorithm.

Even though our approximation factor is a computable function of the treewidth bound, this function grows very rapidly. 
Is it possible to get
a ``reasonable'' approximation factor (say, under 100) for {\sc Weighted Treewidth-2 Deletion} or {\sc Weighted Outerplanar Deletion}?
In the unweighted setting, 
the approximation factor (for both problems) can be pinned down to 40~\cite{DonkersJW22, GuptaLLM019}.
As another potential direction, can we
combine the ideas from \Cref{thm:main:cover} and the scaling lemma from~\cite{FominL0Z20} to
improve the status of \planardel on  $H$-minor-free graphs from QPTAS to PTAS?

One message of our work is that {\sc Planar} \ufdel 
does not get any harder for approximation or FPT algorithms when we consider weighted graphs.
What about a polynomial kernelization under parameterization by the number of vertices in the sought solution?
While this is well-understood on unweighted graphs~\cite{FominLMS12}, we are only aware of a polynomial kernelization for
\wei{ \sc Vertex Cover}~\cite{EtscheidKMR17}. 
This result could be possibly extended to the weighted variants of \fvs, {\sc Treewidth-2 Deletion}, and {\sc Outerplanar Deletion}, for which we know kernelization algorithms based on
explicit reductions rules~\cite{DonkersJW22, schols2022kernelization}.
Does general \planardel admit a polynomial kernelization for every fixed family $\fcal$?

\bibliographystyle{plain}
\bibliography{main}

\iffalse
\appendix
\section{Proof of \Cref{lem:exha:compute}}
\label{app:exha}

We are ready to explain how the exhaustive families are computed. We restate the lemma.

\lemExha*
\fi

\end{document}